\documentclass{jpp}
\usepackage{graphicx}
\usepackage{epstopdf, epsfig}
\usepackage{hyperref}
\usepackage{url}
\usepackage{amsmath}
\usepackage{verbatim}
\usepackage{hyphenat}

\usepackage[mathscr]{euscript}

\usepackage{enumerate}
\usepackage{setspace}
\usepackage{graphicx}
\usepackage{color}
\usepackage{natbib}
\usepackage[bf, footnotesize]{caption2}

\newcommand{\derv}[1]{\frac{\partial}{\partial #1}}

\newcommand{\deriv}[2]{\frac{\partial #1}{\partial #2}}

\newcommand{\beqn}{\begin{equation}}
\newcommand{\eeqn}{\end{equation}}
\newcommand{\beqnar}{\begin{eqnarray}}
\newcommand{\eeqnar}{\end{eqnarray}}

\newtheorem{theorem}{Theorem}[section]

\newtheorem{proposition}[theorem]{Proposition}

\newcommand{\contr}{\,\lrcorner\,}
\clearpage

\shorttitle{Godbillon-Vey Helicity and Magnetic Helicity in MHD}
\shortauthor{G.M. Webb, A. Prasad, S.C. Anco,  and Q. Hu}

\title{Godbillon-Vey Helicity and Magnetic Helicity in Magnetohydrodynamics}

\author{G. M. Webb\aff{1}
  \corresp{\email{gmw0002@uah.edu}},
  A. Prasad\aff{1},
  S.  C. Anco\aff{2},
\& \ Q. Hu\aff{1,3}} 

\affiliation{\aff{1}Center for Space Plasma and Aeronomic Research, the University of Alabama in Huntsville, Huntsville AL 35805, USA.
\aff{2}Department of Mathematics, Brock University, St. Catharines, ON L2S 3A1 Canada.
\aff{3} Department of Space Science, The University of Alabama in Huntsville, 
Huntsville AL35899, USA.}
\begin{document}
\maketitle
\begin{abstract}

The Godbillon-Vey invariant occurs in homology theory,
and algebraic topology, 
when conditions for a co-dimension 1, foliation of a 3D manifold
are satisfied. 
 The magnetic Godbillon-Vey helicity invariant 
in magnetohydrodynamics (MHD) is a higher order helicity invariant 
that occurs for flows, in which the magnetic helicity density $h_m={\bf A}{\bf\cdot}{\bf B}=
{\bf A}{\bf\cdot}(\nabla\times{\bf A})=0$, where ${\bf A}$ is the magnetic vector potential and ${\bf B}$ is the magnetic induction. 
This paper obtains evolution equations for the magnetic Godbillon-Vey 
field $\boldsymbol{\eta}={\bf A}\times{\bf B}/|{\bf A}|^2$ and the Godbillon-Vey 
helicity density $h_{gv}=\boldsymbol{\eta}{\bf\cdot}(\nabla\times{\boldsymbol\eta})$
in general MHD flows in which either $h_m=0$ or $h_m\neq 0$. 
A conservation law for $h_{gv}$ occurs in flows for which $h_m=0$. For 
$h_m\neq 0$  the evolution equation for $h_{gv}$
contains a source term 
 in which $h_m$ is coupled to $h_{gv}$ via the shear tensor 
of the background flow. The transport equation for 
$h_{gv}$ also depends on the electric field potential $\psi$, 
which is related to the gauge for ${\bf A}$, 
which takes its simplest form for the advected ${\bf A}$ gauge 
in which $\psi={\bf A\cdot u}$ 
where ${\bf u}$ is the fluid velocity. An application of the Godbillon-Vey 
magnetic helicity to nonlinear force-free magnetic fields used in 
solar physics is investigated. The possible uses  
of the Godbillon-Vey helicity in zero helicity flows in ideal fluid mechanics,
and in zero helicity Lagrangian kinematics of three-dimensional advection 
are discussed.
\end{abstract}
\section{Introduction}
In ideal fluid dynamics and magnetohydrodynamics (MHD), there is a class of 
invariants that are Lie dragged by the flow (e.g. \cite{Moiseev82}; 
\cite{Tur93}; \cite{Kats03}; \cite{Moffatt69, Moffatt78}; 
\cite{Salmon82, Salmon88};
\cite{Moffatt92}; \cite{Cotter07}; 
\cite{Holm98}; \cite{Padhye96a, Padhye96b}; 
\cite{Yahalom13, Yahalom17a, Yahalom17b}; 
\cite{Webb14a, Webb14b}). These Lie dragged invariants 
in many cases are 
related to fluid relabelling symmetries and Casimirs for non-canonical Hamiltonian brackets
(e.g. \cite{Morrison82}; \cite{Holm83a, Holm83b}; 
\cite{Padhye96a, Padhye96b}; 
\cite{Holm85}; \cite{Morrison98}; \cite{Hameiri04}; \cite{Tanehashi15}; 
\cite{Besse17}). \cite{AncoDar09} have classified conservation laws for compressible isentropic ideal fluids in $n>1$ 
spatial dimensions, and for the case of non-isentropic flows 
in \cite{AncoDar10}. \cite{AncoWebb18} describe heirarchies of vorticity invariants related to conserved helicity and cross helicity integrals for ideal 
fluids, using familiar vector calculus  operations (and their extension 
to tensor calculus). 

Magnetic helicity is an important quantity in MHD describing the magnetic field topology
(e.g. \cite{Elsasser56}; \cite{Woltjer58}; \cite{Kruskal58}; \cite{Berger84}; 
\cite{Finn85, Finn88}; \cite{Moffatt78, Moffatt92};
\cite{Low06, Low11}; \cite{Longcope08}; \cite{Webb10a};  \cite{Bieber87};
\cite{Webb14a, Webb14b}; \cite{Prior14}; \cite{Tanehashi15}; 
\cite{Blackman15}, \cite{Akhmetev17}). 

\cite{Calkin63} and \cite{Webb17} derived the conservation 
law for the magnetic helicity density $h_m={\bf A}{\bf\cdot}{\bf B}$ 
via gauge field theory. The symmetry responsible for the magnetic helicity conservation law, 
for an electric potential $\psi$, where ${\bf E}=-\nabla\psi-\partial{\bf A}/\partial t$ 
and ${\bf B}=\nabla\times {\bf A}$ is not a fluid relabelling symmetry. It is due 
to a gauge symmetry, involving the Lagrange multipliers that enforce Faraday's equation 
and Gauss's equation ($\nabla{\bf\cdot}{\bf B}=0$) in the variational principle 
(\cite{Webb17}). 

In fluid dynamics, the kinetic fluid helicity density 
$h_k={\bf u}{\bf\cdot}(\nabla\times {\bf u})={\bf u}{\bf\cdot}\boldsymbol{\omega}$ 
for a barotropic flow (i.e. the gas pressure: $p=p(\rho)$), satisfies the 
local conservation law:
\begin{equation}
\derv{t}({\bf u}{\bf\cdot}\boldsymbol{\omega})
+\nabla{\bf\cdot}\left[({\bf u}{\bf\cdot}\boldsymbol{\omega}) {\bf u}
+\boldsymbol{\omega}\left(h+\Phi-\frac{1}{2} u^2\right)\right]=0, 
\label{eq:god1.1}
\end{equation}
where $h$ is the gas enthalpy, ${\bf u}$ is the fluid velocity and $\Phi({\bf x})$ 
is an external gravitational potential (e.g. the gravitational 
potential of the Sun for the Solar Wind flow). The conserved integral:
\begin{equation}
H_f=\int_{V_m} {\bf u\cdot\boldsymbol{\omega}}\ d^3x, \label{eq:god1.1a}
\end{equation}
for a volume $V_m$ moving with the fluid is known as the fluid helicity
(e.g.\cite{Moffatt69}).
 If ${\boldsymbol{\omega}\bf \cdot n}=0$ on the 
boundary $\partial V_m$ moving with the flow, then $H_f$ is conserved 
following the flow (e,g, \cite{Moffatt69}), i.e. $dH_f/dt=0$ where 
$d/dt=\partial/\partial t+{\bf u}{\bf\cdot}\nabla$ is the Lagrangian 
time derivative following the flow. The volume integral $H_f$ 
describes the linking of the poloidal and toroidal 
vorticity fluxes. It is used to describe topological 
features of the vortex tubes 
(e.g. whether they are knotted or otherwise).

In ideal MHD, the magnetic helicity conservation law for a 
non-dissipative fluid is given by:
\begin{equation}
\derv{t}\left({\bf A\cdot B}\right)
+\nabla{\bf\cdot}\left[({\bf A\cdot B}){\bf u}
+{\bf B}\left(\psi-{\bf A}{\bf\cdot}{\bf u}\right)\right]=0, \label{eq:god1.2}
\end{equation}
where ${\bf E}=-\nabla\psi-\partial A/\partial t=-({\bf u}\times{\bf B})$ 
is the electric field in the MHD approximation 
and $\psi$ is electric field potential (e.g. \cite{Berger84}). The 
magnetic helicity for a volume $V_m$ moving with the fluid is defined as:
\begin{equation}
H_m({\bf A},{\bf B})=\int_{V_m} {\bf A}{\bf\cdot} {\bf B}\ d^3x, 
\label{eq:god1.2a}
\end{equation}
If ${\bf B\cdot n}=0$ on the boundary $\partial V_m$ then $H_m$ is conserved
moving with the flow, i.e. $dH_m/dt=0$.  
 The helicity integral (\ref{eq:god1.2a})
is independent of the gauge of ${\bf A}$, i.e. 
$H_m({\bf A}+\nabla\psi,{\bf B})=H_m({\bf A},{\bf B})$ 
provided that $\psi$ is smooth and single valued within the volume $V_m$, 
and provided ${\bf B\cdot n}=0$ on the boundary 
$\partial V_m$.   

For magnetic fields in which ${\bf B\cdot n}\neq 0$ on the boundary surface 
$\partial V$, a gauge independent definition of relative helicity 
(\cite{Finn85, Finn88}) is defined as:
\begin{equation}
H_r=\int_V d^3x ({\bf A}+{\bf A}_p){\bf\cdot}({\bf B}-{\bf B}_p), 
\label{eq:god1.2b}  
\end{equation}
(see also \cite{Berger84} for an equivalent definition) 
where ${\bf B}=\nabla\times {\bf A}$ describes the magnetic field of 
interest and ${\bf B}_p=\nabla\times{\bf A}_p$ describes a comparison 
magnetic field, with the same normal flux as ${\bf B}$ on the boundary $\partial V$ 
(in many instances it is useful to choose ${\bf B}_p$ to be a potential 
magnetic field, with the same normal magnetic flux as ${\bf B}$ 
on $\partial V$).

More recent efforts  by \cite{Low06,Low11} and 
\cite{Berger18} discuss  the concept of absolute magnetic helicity  
 which is analogous to the 
 the linkage of the toroidal and poloidal magnetic fluxes. 
 \cite{Kruskal58}  obtained a similar interpretation 
of magnetic helicity for Tokamak fusion devices. The work 
by \cite{Berger18} invokes the Gauss-Bonnet theorem as 
part of the discussion and does not at the outset assume 
that the field splits cleanly into toroidal and poloidal 
components. 
 
 There are other conservation laws in 
MHD. In particular, the cross helicity  
density $h_c={\bf u}{\bf\cdot}{\bf B}$  conservation law for barotropic flows 
 is important in MHD turbulence theory (e.g. \cite{Zhou90a, Zhou90b}; 
\cite{Zank12}) and in MHD (e.g. \cite{Webb14a, Webb14b}). 
The cross helicity integral is defined as $H_c=\int_{V_m} {\bf u\cdot B}\ d^3x$
where ${\bf B\cdot n}=0$ on $\partial V_m$. In ideal barotropic MHD 
$dH_c/dt=0$. A generalized, nonlocal  cross helicity applies for non-barotropic 
MHD (e.g. \cite{Webb14a, Webb14b}, \cite{Yahalom17a, Yahalom17b}).
 Cross helicity describes the linkage 
of the vortex tubes and magnetic flux tubes. This definition of cross helicity
is that conventionally used in plasma physics, but it has a wider definition 
in terms of the cross helicity 
density ${\bf V}{\bf\cdot}(\nabla\times{\bf W})$ for two 
vector fields ${\bf V}$ and ${\bf W}$. 
\cite{Yahalom13, Yahalom17a, Yahalom17b}
has described magnetic helicity, barotropic cross helicity 
and nonlocal (non-barotropic) cross helicity
in terms of MHD Aharonov-Bohm effects. 

\cite{Tur93}, \cite{Webb14a}, \cite{Webb18}, 
 and \cite{AncoWebb18} give discussions of Lie dragged 
vector fields $\mathbf{b}\contr \partial/\partial\mathbf{x}$, one forms 
$\boldsymbol{\omega}= {\bf C}{\bf\cdot} d{\bf x}$, two forms ${\bf J}{\bf\cdot}d{\bf S}$ and three forms $D d^3x$
and scalars, $R$. 
An example of a Lie dragged two-form in MHD is the magnetic flux 
$\beta={\bf B}{\bf\cdot}d{\bf S}$. Faraday's equation can then be expressed 
in terms of the Lie dragging of the two-form $\beta$ with the flow (i.e. 
Faraday's law is equivalent to the statement that the magnetic flux $\beta$
is conserved moving with the flow). Entropy $S$ is an advected scalar, 
and $[\mathbf{B}/\rho]\contr\partial/\partial\mathbf{x}$  
is an invariant  vector 
field which is Lie dragged with the flow. 

\cite{Tur93} in their study of Lie dragged invariants in MHD flows, 
asked the question: Given ${\bf A}{\bf\cdot}{\bf B}=0$, is there 
a higher order magnetic,  Lie dragged 
integral invariant (i.e. volume integral analogous to $H_m$,  
e.g. \cite{Tur93} and \cite{Webb18}). 
The answer to this question is that in general, 
there is a higher order topological invariant
known as the Godbillon-Vey invariant for 
flows with ${\bf A}{\bf\cdot}{\bf B}=0$.
The condition ${\bf A}{\bf\cdot}{\bf B}\equiv {\bf A}{\bf\cdot}(\nabla\times{\bf A})=0$ 
is the condition that the Pfaffian equation ${\bf A}{\bf\cdot}d {\bf x}=0$ is integrable 
(e.g. \cite{Sneddon57}, Ch. 1). The Pfaffian is integrable means that there exists an integrating 
factor $\mu$ such that $\mu {\bf A} {\bf\cdot}d{\bf x}=\nabla 
\lambda{\bf\cdot} d{\bf x}=d\lambda$ 
for some potential function 
$\lambda$. In this case, the family of surfaces $\lambda({\bf x})=c_1=const.$ 
forms a solution family of 2D 
surfaces (a foliation) with normal 
$\hat{\bf A}={\bf A}/|{\bf A}|\propto \nabla\lambda$ 
which fill up 3D space 
(see  \cite{Sneddon57} for a proof of both the necessity and sufficiency of the 
condition for integrability). This same idea was used by \cite{Godbillon71}  
to describe foliations of co-dimension 1, in 3D space, in homology theory. Homology theory 
has wide applications in algebraic topology, which is concerned with the genus (number of holes 
in a surface) and other topological invariants in the geometry of manifolds 
(e.g. \cite{Thurston72},
\cite{Arnold98}, \cite{Fulton95}, \cite{Lee97}).  

The Godbillon-Vey one-form $\boldsymbol{\eta}{\bf\cdot}d{\bf x}$ 
and the Godbillon Vey helicity 3-form: 
$\boldsymbol{\eta}{\bf\cdot} (\nabla\times\boldsymbol{\eta})\ d^3x$, are 
also defined for flows with 
${\bf A}{\bf\cdot}{\bf B}\neq 0$ (but in that case the space 
does not consist of a family 
of 2D surfaces filling up 3D space). The integral form of the Godbillon 
Vey helicity for a finite volume $V_m$ moving with the flow is defined as
\begin{equation}
H_{gv}=\int_{V_m}\boldsymbol{\eta}{\bf\cdot}(\nabla\times\boldsymbol{\eta})
\ d^3x  
\quad\hbox{where}\quad \boldsymbol{\eta}=\frac{{\bf A}\times {\bf B}}{|{\bf A}|^2},
\label{eq:1.2c}
\end{equation} 
is  the Godbillon-Vey vector field 
(\cite{Godbillon71, Reinhart73}).  
If ${\bf B}{\bf\cdot n}=0$ on $\partial V_m$ 
and if ${\bf A\cdot B}=0$, the Godbillon-Vey helicity $H_{gv}$ is 
conserved following the flow, i.e. $dH_{gv}/dt=0$. This result is not true 
if ${\bf A\cdot B}\neq 0$.

The Godbillon-Vey helicity studied by \cite{Reinhart73} corresponds 
to using a unit vector for ${\bf A}$, $\hat{\bf A}=\bf{A}/|{\bf A}|$, 
and the Godbillon-Vey field is given by $\hat{\boldsymbol\eta}=\hat{\bf A}
\times(\nabla\times{\hat{\bf A}})$ and the Godbillon-Vey helicity density 
is given by 
$\hat{\boldsymbol\eta}{\bf\cdot}\nabla\times\hat{\boldsymbol\eta}$ 
(see also discussion in Appendix E). The \cite{Reinhart73} 
meaning of $\hat{\bf A}$ is just the unit normal to the foliation, 
and does not have any connection to MHD.
  
For the one-form $\alpha={\bf A}{\bf\cdot}d{\bf x}$, the Reeb vector field 
${\bf R}$
 satisfies ${\bf R}\lrcorner\alpha=1$ 
and ${\bf R}\lrcorner (d\alpha)=0$. Because $d\alpha={\bf B}{\bf\cdot}d{\bf S}
=B_x dy\wedge dz+B_y dz\wedge dx+B_z dx\wedge dy$ we require that 
${\bf R}\lrcorner ({\bf B}{\bf\cdot}d{\bf S})=
-{\bf R}\times{\bf B}{\bf\cdot}d{\bf x}=0$. Thus, the two 
conditions for the Reeb vector are that ${\bf R}{\bf\cdot}{\bf A}=1$ 
and ${\bf R}\times {\bf B}=0$. One solution of the above equations is 
${\bf B}=\beta {\bf A}$ and ${\bf R}=\lambda {\bf B}=\lambda\beta {\bf A}$. 
These conditions lead to the equation $\nabla\times{\bf A}=\beta {\bf A}$ where 
$\beta={\bf A\cdot B}/A^2$ and $\lambda =1/{\bf A\cdot B}$. The equation 
for ${\bf A}$ is that for a Beltrami flow, i. e. the Reeb vector 
${\bf R}$ corresponds to a Beltrami flow. 
 The MHD topological soliton (\cite{Kamchatnov81} 
and \cite{Semenov02}) satisfies $\nabla\times {\bf A}=\beta {\bf A}$ 
where $\beta=k A$ and $k$ is a constant. Similarly, the well known ABC flow
(Arnold, Beltrami, Childress flows) studied by \cite{Dombre86} are examples
of Beltrami flows, which exhibit both chaotic and integrable flows. 

Force free magnetic fields satisfying $\nabla\times{\bf B}=\alpha {\bf B}$
are Beltrami fields which are used to model solar 
magnetic field structures in highly conducting, low beta photospheric 
 plasmas (e.g. \cite{Chandrasekhar57}, \cite{Low90}, \cite{Prasad14}).
\cite{Prasad14} have shown that the \cite{Low90}  force free magnetic fields 
have zero magnetic helicity $h_m={\bf A\cdot B}$ in an appropriate gauge. 
This class of fields are clearly examples of magnetic fields 
that can in principle have a non-zero Godbillon-Vey helicity, 
but have zero helicity in the gauge used by \cite{Prasad14}. 
\cite{Prasad14} show that the \cite{Low90} solutions 
have non-trivial relative magnetic helicity. 

The aim of the present paper is to  derive 
an evolution equation for the Godbillon-Vey helicity density
$h_{gv}= \boldsymbol{\eta}{\bf\cdot}(\nabla\times\boldsymbol{\eta})$, 
for general MHD flows, 
in which $h_m={\bf A}{\bf\cdot}{\bf B}\neq 0$.  We show, that there is 
an intimate connection 
between the Godbillon-Vey helicity $h_{gv}$ evolution and the 
magnetic helicity density 
$h_m={\bf A}{\bf\cdot}{\bf B}$ in which $h_m$ acts as a source in the $h_{gv}$ equation, 
in which the shear tensor of the flow, acts as a coupling agent between the two types of 
helicity. 

In Section 2 we introduce the usual MHD equations and the magnetic helicity transport equation 
derived by \cite{Berger84} and others. 
In Section 3 we derive (a)\ the magnetic helicity transport equation 
and (b)\ describe the 
Godbillon-Vey one-form and helicity. In Section 4, we derive the transport equation for the 
Godbillon-Vey helicity $h_{gv}$ based on a decomposition of the magnetic field induction 
${\bf B}$ in the form:
\begin{equation}
{\bf B}={\bf B}_{\parallel}+{\bf B}_\perp=\beta{\bf A}+ \boldsymbol{\eta}\times {\bf A}, 
\label{eq:god1.3}
\end{equation}
where
\begin{equation}
\boldsymbol{\eta}=\frac{{\bf A}\times{\bf B}}{|{\bf A}|^2},\quad 
\beta=\frac{h_m}{|{\bf A}|^2}, \quad\hbox{and}\quad h_m={\bf A}{\bf\cdot}{\bf B}. 
\label{eq:god1.4}
\end{equation}
Equation (\ref{eq:god1.3}) can also be written in the form:
\begin{equation}
{\bf B}_\parallel={\bf B}{\bf\cdot}\hat{\bf A}\hat{\bf A}\equiv\beta {\bf A},
\quad {\bf B}_\perp={\bf B}-{\bf B}{\bf\cdot}\hat{\bf A}\hat{\bf A}\equiv
\boldsymbol{\eta}\times{\bf A}, \label{eq:god1.5}
\end{equation}
are the components of ${\bf B}$ parallel and perpendicular to ${\bf A}$, and 
$\hat{\bf A}={\bf A}/|{\bf A}|$ is the unit vector parallel to ${\bf A}$.

Section 5 determines the Godbilllon-Vey magnetic helicity density 
for the \cite{Low90} nonlinear, force-free magnetic fields used to describe 
photospheric magnetic fields in solar physics. 

Section 6 concludes with a summary and discussion.

In appendix A, we provide a detailed derivation of the conservation law  for 
the Godbillon helicity density 
$h_{gv}=\boldsymbol{\eta}{\bf\cdot}(\nabla\times\boldsymbol{\eta})$ 
for the case ${\bf A\cdot B}=0$ using the Lie dragging of differential 
forms (see also \cite{Tur93}, \cite{Webb14a}, \cite{Webb18}). 
Appendix B, provides a vector Calculus derivation of the Godbillon-Vey 
helicity evolution equation for general MHD flows, 
both for the case ${\bf A\cdot B}=0$ and  for the case
${\bf A\cdot B}\neq 0$ (We also discuss  
the gauge potential used for ${\bf A}$). In Appendix C, 
we explore the use of Clebsch potential representations for ${\bf A}$
which are related to the integrability of ${\bf A}{\bf\cdot}d{\bf x}$
in the case ${\bf A\cdot B}=0$. 
We obtain the form of $h_{gv}$ in terms
of Clebsch variables or Euler potentials, which are advected with the flow.  
The analysis in Appendix C can be further developed using               
the differential geometry  of surfaces in three space dimensions
 (e.g. \cite{Lipschutz69}, \cite{Boozer83, Boozer04}, \cite{Kobayashi63}, 
\cite{Lee97}). 
Appendix D discusses gauge transformations for the magnetic vector 
potential ${\bf A}$ which are compatible with 
the condition ${\bf A}{\bf\cdot} {\bf B}=0$ and co-dimension one foliations. 
Appendix E derives the \cite{Reinhart73} 
 formula for the Godbillon-Vey invariant for a co-dimension 1 foliation in three-dimensional geometry (i.e. 
a family of two dimensional surfaces or foliation), in terms of the 
curvature and torsion of the curves normal to the foliation, 
and in terms of the 
second fundamental form for the surface. The connection between  
the differential geometry formulation of the Godbillon-Vey invariant 
by \cite{Reinhart73} 
and the Godbillon-Vey invariant used in this paper is described. 
Appendix F describes Clebsch potential representations for the \cite{Low90} 
nonlinear force free magnetic fields. Appendix G describes 
the \cite{Reinhart73} form of the Godbillon-Vey invariant for the 
\cite{Low90} force-free magnetic field using spherical polar coordinates.
\section{The MHD Equations}

The ideal MHD equations, consist of the mass continuity equations:
\begin{equation}
\deriv{\rho}{t}+\nabla{\bf\cdot}(\rho {\bf u})=0; \label{eq:god2.1}
\end{equation}
the momentum equation:
\begin{equation}
\derv{t}\left(\rho {\bf u}\right)+\nabla{\bf\cdot}
\biggl[\rho {\bf u}{\bf u}
+\left(p+\frac{B^2}{2\mu_0}\right){\sf I}-\frac{{\bf B}{\bf B}}
{\mu_0}\biggr]=0; \label{eq:god2.2}
\end{equation}
the entropy advection equation:
\begin{equation}
\deriv{S}{t}+{\bf u}{\bf\cdot}\nabla S=0; \label{eq:god2.3}
\end{equation}
Faraday's equation:
\begin{equation}
\deriv{\bf B}{t}-\nabla\times({\bf u}\times{\bf B})=0; 
\label{eq:god2.4}
\end{equation}
and Gauss's equation:
\begin{equation}
\nabla{\bf\cdot}{\bf B}=0; \label{eq:god2.5}
\end{equation}
supplemented by the first law of thermodynamics, 
which is related to the equation 
of state for the gas in ideal MHD (e.g. $p=p(\rho,S)$). Here $\rho$, ${\bf u}$, $p$, $S$, 
and ${\bf B}$  are the gas density, fluid velocity, 
pressure, entropy and magnetic field 
induction respectively.  
 Faraday's equation (\ref{eq:god2.4}) is sometimes written with the 
addition of an  
extra term of ${\bf u}\nabla{\bf\cdot}{\bf B}$ 
on the left-hand side. This is useful in numerical MHD, 
where numerically generated $\nabla{\bf\cdot}{\bf B}\neq 0$ 
can cause numerical errors and instabilities in the MHD system.  
The problem of the effects of $\nabla{\bf\cdot}{\bf B}\neq 0$, and the methods
used to reduce numerically generated $\nabla{\bf\cdot}{\bf B}$ have been 
extensively discussed in the numerical MHD literature (e.g. \cite{Evans88},
\cite{Powell99},
\cite{Janhunen00},  \cite{Dedner02},\cite{Balsara04}, \cite{Stone09}, 
\cite{Webb10}).

Because $\nabla{\bf\cdot B}=0$ (Gauss's equation), we can express ${\bf B}$
in terms of the magnetic vector potential ${\bf A}$ as: 
\begin{equation}
{\bf B}=\nabla\times{\bf A}. \label{eq:god2.6}
\end{equation}
Faraday's equation (with $\nabla{\bf\cdot}{\bf B}=0$) in ideal MHD 
is given by:
\begin{equation}
\deriv{\bf B}{t}+\nabla\times{\bf E}=0
\quad\hbox{where}\quad {\bf E}=-{\bf u}\times{\bf B}, 
\label{eq:god2.9}
\end{equation}
is the electric field in the fixed inertial frame.   
From (\ref{eq:god2.6})-(\ref{eq:god2.9}),  
\begin{equation}
\nabla\times({\bf A}_t+{\bf E})=0, \label{eq:god2.9a}
\end{equation}
 implying:
\begin{equation}
{\bf E}=-\nabla\psi-\deriv{\bf A}{t}\quad \hbox{or}\quad \deriv{\bf A}{t}+{\bf E}+\nabla\psi=0. 
\label{eq:god2.10}
\end{equation}
Here $\psi$ is an arbitrary gauge potential obtained by 
solving (\ref{eq:god2.9a}) for ${\bf E}$. 
Equations (\ref{eq:god2.9})-(\ref{eq:god2.10}) 
and Gauss's equation $\nabla{\bf\cdot B}=0$ are used below to derive 
the local conservation law for the 
magnetic helicity density $h_m={\bf A\cdot B}$.

\section{Magnetic helicity and Godbillon-Vey invariant}
In this section we derive the magnetic helicity transport equation, 
and the Godbillon-Vey helicity transport equation.

\subsection{\bf Magnetic helicity}
Using the two forms of Faraday's equation (\ref{eq:god2.9}) 
and (\ref{eq:god2.10}) 
in the combination:
\begin{equation}
{\bf A}{\bf\cdot}\left({\bf B}_t+\nabla\times {\bf E}\right) 
+{\bf B\cdot}\left({\bf A}_t+{\bf E}+\nabla\psi\right)=0, \label{god3.1}
\end{equation}
results in the magnetic helicity transport equation:
\begin{equation}
\derv{t}\left({\bf A\cdot B}\right) 
+\nabla{\bf\cdot}\left({\bf E}\times {\bf A}+\psi {\bf B}\right)=-2 {\bf E}{\bf\cdot B}. 
\label{eq:god3.2}
\end{equation}
In ideal MHD, ${\bf E}{\bf\cdot}{\bf B}=-({\bf u}\times{\bf B}){\bf\cdot}{\bf B}=0$, 
and in this limit, (\ref{eq:god3.2}) reduces to the magnetic helicity conservation 
equation:
\begin{equation} 
\derv{t}\left({\bf A}{\bf\cdot}{\bf B}\right)
+\nabla{\bf\cdot}\left[{\bf u}\left({\bf A}{\bf\cdot}{\bf B}\right) 
+\left(\psi-{\bf A}{\bf\cdot}{\bf u}\right) {\bf B}\right]=0. \label{eq:god3.3}
\end{equation}

For the case of a non-ideal plasma, with finite conductivity $\sigma$, 
the simplest form of Ohm's  law for the plasma has the form:
\begin{equation}
{\bf E}=-{\bf u}\times{\bf B}+\frac{\bf J}{\sigma}\quad \hbox{or}\quad 
{\bf E}'={\bf E}+{\bf u}\times{\bf B}= \frac{\bf J}{\sigma}, 
\label{eq:god3.4}
\end{equation}
in which ${\bf E}'$ is the electric field in the fluid frame (e.g. Boyd and Sanderson (1969), 
equation (3.61)). The magnetic helicity transport equation (\ref{eq:god3.2}) 
reduces to the equation:
\begin{equation}
\derv{t}\left({\bf A\cdot B}\right)
+\nabla{\bf\cdot}\left[{\bf u}\left({\bf A\cdot B}\right)
+\left(\psi-{\bf A\cdot u}\right) {\bf B}
+\frac{{\bf J}\times {\bf A}}{\sigma}\right]
= -\frac{2({\bf J\cdot B})}{\sigma}. 
\label{eq:god3.5}
\end{equation}
By integrating (\ref{eq:god3.5}) over a volume $V_m$ moving with the flow
gives the equation:
\begin{equation}
\frac{dH_m}{dt}=\int _{\partial V_m} {\bf B\cdot n}\left({\bf A}{\bf\cdot u}
-\psi\right)\ dS
-\int_{\partial V_m} \frac{{\bf n\bf\cdot}({\bf J}\times {\bf A})}{\sigma}\ dS
-\int_{V_m} 2\frac{{\bf J\cdot B}}{\sigma}\ d^3 x.  \label{eq:god3.5a}
\end{equation}
The surface term involving ${\bf B\cdot n}$ vanishes as ${\bf B\cdot n}=0$
is assumed on $\partial V_m$. The second term represents the transport of 
helicity flux across $\partial V_m$ and the volume integral represents 
 dissipation of the helicity  due to the finite conductivity 
of the plasma. 

 \cite {Taylor86} developed a theory 
for the decay of magnetic helicity in a high conductivity plasma 
 by hypothesizing that at lowest order  
the magnetic helicity for the whole volume $V_m$  
is conserved, but locally there could be magnetic reconnection of the field 
${\bf B}$. Taylor's theory leads to a much faster decay rate for the magnetic 
energy density of the field in a weakly dissipative plasma than 
for the magnetic helicity. 
-

 The total magnetic helicity for a volume $V_m$ 
moving with the flow is defined as the integral:
\begin{equation}
H_m=\int_{V_m} {\bf A\cdot B}\ d^3x. 
\label{eq:god3.6}
\end{equation}
In the ideal MHD limit ($\sigma\to\infty$) $H_m$ 
is conserved following the flow,  
i.e. $dH_m/dt=0$, provided ${\bf B}{\bf\cdot}{\bf n}=0$
on the boundary surface $\partial V_m$ of the volume $V_m$. 

The magnetic helicity integral (\ref{eq:god3.6}) can be written as:
\begin{equation}
H_m= \int_{V_m} \boldsymbol{\omega}_A^1\wedge d\boldsymbol{\omega}_A^1,
\label{eq:god3.6a}
\end{equation}
where
\begin{equation}
\boldsymbol{\omega}_A^1={\bf A}{\bf\cdot} d{\bf x}\quad \hbox{and} 
\quad \boldsymbol{\omega}_A^2=d \boldsymbol{\omega}_A^1
={\bf B}{\bf\cdot} d{\bf S} \label{eq:god3.7}
\end{equation}
are the magnetic vector potential one-form $\boldsymbol{\omega}_A^1$ and 
the magnetic flux two-form 
$\boldsymbol{\omega}_A^2=d\boldsymbol{\omega}_A$.  
The symbol $\wedge$ denotes the 
wedge product 
used in the algebra of exterior differential forms (e.g. \cite{Flanders63}). 
The integral form (\ref{eq:god3.6a}) is known as the Hopf invariant
which was developed in topological field theory by Hopf in the 1930's.

The proof that $dH_m/dt=0$ for the case where the volume consists 
of flux tubes, in which ${\bf B\cdot n}=0$ on $\partial V_m$ was derived by 
\cite{Moffatt78} 
(see also \cite{Woltjer58} and \cite{Elsasser56} for more discussion). 

\subsection{\bf The Godbillon-Vey invariant}
\begin{figure}[!htb]
\vspace*{1cm}
\centering\includegraphics[width=4cm, angle=0]{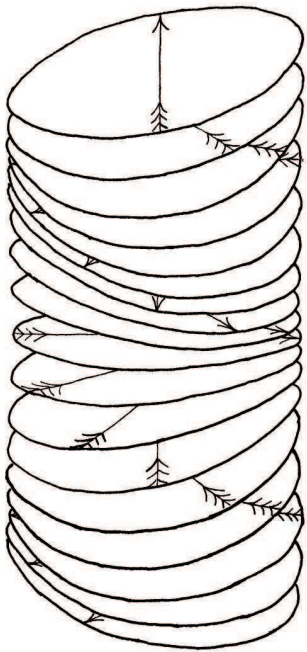}
\caption{Illustrating the wobble of the foliations described
by the Godbillon-Vey invariant (from Thurston (1972)
AMS Vol. 76 (4), July 1972, 511-514.)}
\label{fig:gv-wobble}
\end{figure}

The Godbillon-Vey invariant was introduced by \cite{Godbillon71} 
and later studied 
by \cite{Reinhart73}, \cite{Hurder02} and others. In 3D geometry, 
one can imagine the space as being filled with a family of 2 dimensional 
surfaces in which the surfaces are solutions of the Pfaffian 
equation $\boldsymbol{\omega}_A^1={\bf A}{\bf\cdot}d{\bf x}=0$. 
In the present paper the Godbillon-Vey field is defined as:
\begin{equation}
\boldsymbol{\eta}={\bf A\times {\bf B}/|{\bf A}|^2}. \label{eq:god3.7a}
\end{equation}
The reason for this choice for $\boldsymbol{\eta}$ is outlined below.
 The Godbillon-Vey 3-form is  
the three-form:
\begin{equation}
\boldsymbol{\omega}_\eta^3=
\boldsymbol{\omega}_\eta^1 \wedge d\boldsymbol{\omega}_\eta^1
\equiv \boldsymbol{\eta}{\bf\cdot}(\nabla\times \boldsymbol{\eta}) 
\ d^3x, \label{eq:god3.8}
\end{equation}
where $\boldsymbol{\eta}$ is the Godbillon-Vey field. 
The Godbillon -Vey invariant for the volume $V_m$ is the helicity integral
\begin{equation}
H_{gv}=\int_{V_m}\boldsymbol{\omega}^1_\eta\wedge d\boldsymbol{\omega}^1_\eta, 
\label{eq:god3.8a}
\end{equation}
\cite{Thurston72} described the Godbillon-Vey invariant as 
the wobble of a foliation or a pyramid of discs
lying on top of each other (see Figure \ref{fig:gv-wobble}). 
A similar wobble 
can be seen in the sculpture illustrated in  
Figure \ref{fig:gv-wobble2} (from \cite{Ghys14} lecture on `Foliations: 
What's next 
after Thurston'). 

The  meaning of `wobble' used above is 
presumably  related to the wobble of spinning objects,  
due to the tilt between the axis of symmetry and its angular momentum 
(this in solid body dynamics involves the moments of inertia of the spinning
body and the rotation axis of spin). A description of this phenomenon
for rigid bodies is quite complicated 
(see e.g. \cite{Goldstein80}; \cite{Holm08}, \cite{Marsden94}, Chapter 15, and
also the webpage http:/www/mathpages.com/home/kmath/kmath116.htm). 

\subsection{\bf \cite{Reinhart73} formula for Godbillon-Vey invariant}

\cite{Reinhart73} (see also appendix E) show that the Godbillon-Vey invariant
can be written in the form:
\begin{equation}
H_{gv}^{RW}=\int_{V_g} \hat{\eta}\wedge d \hat{\eta} 
=\int_{V_g} \kappa^2 \left(\tau-h_{BN}\right) d^3x, 
\label{eq:rw1}
\end{equation}
where $\kappa$ and $\tau$ are the curvature and torsion of a curve 
(or family of curves) with tangent vector ${\bf T}=\hat{\bf A}$ normal to the 
foliation (here $\hat{\boldsymbol\eta}=\hat{\bf A}\times(\nabla\times
\hat{\bf A})=
-\hat{\bf A}{\bf\cdot}\nabla\hat{\bf A}=-{\bf k}$ 
where ${\bf k}$ is the curvature vector of the curve.
 $h_{ij}$ ($i,j=1,2$) defines the second fundamental form for 
the surface, which describes the curvature of the foliation surface 
$\Phi=const.$.  Here ${\bf T},{\bf N}, {\bf B}$ is the moving tri-hedron 
for the curve normal to the foliation, with tangent vector ${\bf T}$, 
principal normal ${\bf N}$ and bi-normal ${\bf B}$. These  vector fields 
are governed by the Serret-Frenet formulae 
(or suitable equivalent formulae), and we use the notation ${\bf T}={\bf e}_3$,
${\bf N}={\bf e}_1$ and ${\bf B}={\bf e}_2$ for the orthonormal moving 
tri-hedron ${\bf T},{\bf N}, {\bf B}$,  so that $h_{21}=h_{BN}$. 
The main point of the \cite{Reinhart73} formula is that the 
curve with tangent vector ${\bf T}$ normal to the surface is a non-planar curve
as it has both non-zero curvature ($\kappa$) and torsion ($\tau$), 
and it in 
general wobbles out of the original plane of the curve (for example if $h_{BN}$
is small then both $\kappa$ and $\tau$ must be non-zero  
in order for the differential invariant $\hat{\eta}\wedge d\hat{\eta}$
to be non-zero).  
\begin{figure}
\centering\includegraphics[width=6cm, angle=0]
{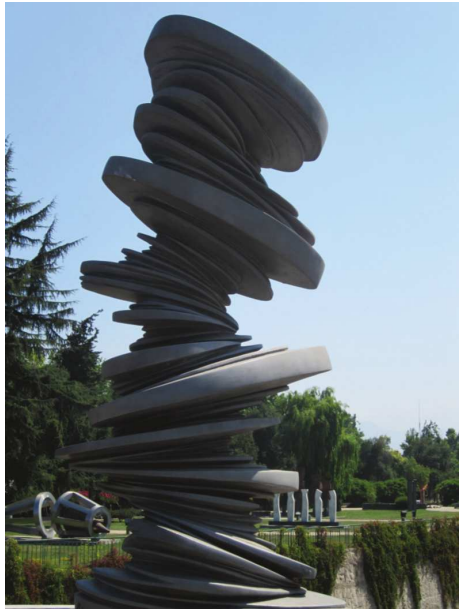}
\caption{Illustrating the wobble of the foliations 
 for a sculpture 
by Alejandra Ruddoff: (Diacronia, 2005), used by Ghys (2014)
in his lecture 
 on foliations ( \cite{Ghys14})}\label{fig:gv-wobble2}
\end{figure}
\medskip

   \cite{Godbillon71} 
and \cite{Reinhart73}, describe  
the Godbillon-Vey invariant 
 for a  co-dimension 1 foliation as a co-homology class defined 
by the 3-form (\ref{eq:god3.8}).  
 This theory is important in algebraic topology
in describing the topology of the distinct classes of closed 
curves that can be drawn on 
 hypersurfaces in terms of the so-called Betti numbers and other 
topological invariants 
(see e.g. \cite{Fulton95}, \cite{Hatcher02}, 
\cite{Pontryagin52} describes simplexes and combinatorial topology). 
\subsection{\bf The MHD Godbillon-Vey Field and Invariant}

The reason for the 
choice of $\boldsymbol{\eta}$ in (\ref{eq:god3.7a}) 
is given below. 
The condition for the Pfaffian 
equation $\boldsymbol{\omega}_A^1={\bf A}{\bf\cdot}d{\bf x}=0$ to be integrable 
 defines a co-dimension 1 foliation, is that:
\begin{equation}
{\bf A}{\bf\cdot}(\nabla\times{\bf A})\equiv {\bf A}{\bf\cdot}{\bf B}=0. \label{eq:god3.9}
\end{equation}
In this case, the Pfaffian equation 
$\boldsymbol{\omega}_A^1={\bf A}{\bf\cdot}d{\bf x}=0$ has an integrating factor $\mu$
such that $\mu {\bf A}=\nabla \lambda$ in which the foliation is described by the family 
of surfaces $\lambda(x,y,z)=c_1$. Each member of the family has unit normal 
$\hat{\bf A}={\bf A}/|{\bf A}|^2$ (i.e. the normal to the surfaces are 
parallel to ${\bf A}$). The integrability condition (\ref{eq:god3.9})
can be expressed as:
\begin{equation}
\boldsymbol{\omega}_A^1 \wedge d \boldsymbol{\omega}_A^1\equiv {\bf A\cdot B}\ d^3x=0. 
\label{eq:god3.10}
\end{equation}
The condition (\ref{eq:god3.10}) is satisfied if there exists a 
1-form:
\begin{equation}
 \boldsymbol{\omega}_\eta^1=\boldsymbol{\eta}{\bf\cdot} d{\bf x}\quad\hbox{such that}\quad 
d\boldsymbol{\omega}_A^1 = \boldsymbol{\omega}_\eta^1\wedge \boldsymbol{\omega}_A^1. 
\label{eq:god3.11}
\end{equation}
In that case,
\begin{equation}
\boldsymbol{\omega}_A^1\wedge d\boldsymbol{\omega}_A^1=\boldsymbol{\omega}_A^1\wedge 
(\omega_\eta^1\wedge\boldsymbol{\omega}_A^1)=0. \label{eq:god3.12}
\end{equation}
Condition (\ref{eq:god3.11}) can be written as:
\begin{equation}
{\bf B}{\bf\cdot} d{\bf S}=(\nabla\times{\bf A}){\bf\cdot} d{\bf S}=(\boldsymbol{\eta}{\bf\cdot}d{\bf x})
\wedge ({\bf A\cdot} d{\bf x})
=\left(\boldsymbol{\eta}\times{\bf A}\right){\bf\cdot}d{\bf S}. 
\label{eq:god3.13}
\end{equation}
Equation (\ref{eq:god3.13}) is satisfied if 
\begin{equation}
{\bf B}_\perp=\boldsymbol{\eta}\times{\bf A}\quad\hbox{where}\quad {\bf B}_\parallel= {\bf B}{\bf\cdot}\hat{\bf A}\hat{\bf A}\quad\hbox{and}\quad {\bf B}_\perp
={\bf B}-{\bf B}{\bf\cdot} \hat{\bf A}\hat{\bf A}, \label{eq:god3.14}
\end{equation}
where the subscripts $\parallel$ and $\perp$ denote components of ${\bf B}$ 
parallel and perpendicular to ${\bf A}$. Taking the cross product of ${\bf A}$ 
on the left with (\ref{eq:god3.14}) gives:
\begin{equation}
{\bf A}\times {\bf B}_\perp={\bf A}\times{\bf B}={\bf A}\times(\boldsymbol{\eta}\times{\bf A})
=({\bf A\cdot A})\boldsymbol{\eta}-({\bf A\cdot}\boldsymbol{\eta}) {\bf A}, 
\label{eq:god3.15}
\end{equation}

Choosing $\boldsymbol{\eta}$ such that $\boldsymbol{\eta}{\bf\cdot}{\bf A}=0$, 
(\ref{eq:god3.15}) gives:
\begin{equation}
\boldsymbol{\eta}=\frac{{\bf A}\times{\bf B}}{|\boldsymbol{A}|^2}. \label{eq:god3.16}
\end{equation}
This is the form of the Godbillon-Vey field that was adopted by \cite{Tur93}
and \cite{Webb14a}. From (\ref{eq:god3.14}) we obtain:
\begin{equation}
{\bf B}=\beta {\bf A}+\eta\times{\bf A}={\bf B}_\parallel+{\bf B}_\perp, \label{eq:god3.17}
\end{equation}
where
\begin{equation}
\beta=\frac{\bf A\cdot B}{|\boldsymbol{A}|^2}=\frac{h_m}{|\boldsymbol{A}|^2}, 
\quad h_m={\bf A\cdot B}. \label{eq:god3.18}
\end{equation} 
The formulas (\ref{eq:god3.16})-(\ref{eq:god3.18}) play an essential role in the formulation 
of a transport equation for the Godbillon-Vey magnetic helicity for both the cases $h_m=0$ and
$h_m\neq 0$.  

It is interesting to note that:
\begin{equation}
{\bf A}\times{\bf B}={\bf A}\times(\nabla\times {\bf A}) 
=\nabla\left(\frac{1}{2}|{\bf A}|^2\right)-{\bf A}{\bf\cdot}\nabla {\bf A}. 
\label{eq:god3.19}
\end{equation}
This result is analogous to the ${\bf J}\times {\bf B}$ force
on the plasma, except that ${\bf B}$ has been replaced by ${\bf A}$
and there is a sign change. The first term is analogous to the gradient of 
a uniform pressure gradient of $A^2$ and the second term is analogous 
to the tension force of the magnetic field in the ${\bf J}\times{\bf B}$
force.  

However, if we use normalized base vectors (i.e. unit vectors 
$\hat{\bf A}={\bf A}/A$)
 to describe the field, then we obtain:
\begin{equation}
\boldsymbol{\eta}=\frac{{\bf A}\times{\bf B}}{|{\bf A}|^2}=\frac{A\nabla A}{A^2}
-\frac{A\hat{\bf A}{\bf\cdot}\nabla(A\hat{\bf A})}{A^2}
=({\sf I}-\hat{\bf A}\hat{\bf A}){\bf\cdot}\nabla[\ln(A)]-
\hat{\bf A}{\bf\cdot}\nabla\hat{\bf A}, \label{eq:god3.20}
\end{equation}
The first term in (\ref{eq:god3.20}) is the gradient of 
$\ln(A)$ perpendicular to ${\bf A}$ and the second term is minus 
the curvature vector of $\hat{\bf A}$. 
The ${\bf A}$ field line curvature term 
in (\ref{eq:god3.20}) can be written in the form:
\begin{equation}
-\hat{\bf A}{\bf\cdot}\nabla\hat{\bf A}=-\kappa^{(A)} {\bf n}^{(A)}, 
\label{eq:god3.21}
\end{equation}
where ${\bf n}^{(A)}$ is the principal normal to the ${\bf A}$ field lines, 
and $\kappa^{(A)}$ is the curvature of the ${\bf A}$ field lines. 
 $\hat{\bf A}$ can be thought of as the tangent vector to a curve with 
principal normal ${\bf n}^{(A)}$ pointed towards the center of 
curvature of the $\hat{\bf A}$ field 
(e.g. \cite{Lipschutz69}). 


\subsection{\bf Gauge Transformations}

For the case ${\bf A\cdot B}=0$, (\ref{eq:god3.20}) for $\boldsymbol{\eta}$ 
implies $\boldsymbol{\eta}{\bf\cdot A}=\boldsymbol{\eta}{\bf\cdot B}=0$. Thus, 
${\bf A}$, $\bf {B}$ and $\boldsymbol{\eta}$ are  
mutually orthogonal vectors, in which ${\bf A}$ is normal to the foliation 
$\lambda=const.$. Note that  
\begin{equation} 
\mu {\bf A}{\bf\cdot}d{\bf x}=d\lambda=\nabla\lambda{\bf\cdot}d{\bf x}. 
\label{eq:god3.23}
\end{equation}
where $\mu$ is an integrating factor.  
It is necessary 
to keep in mind that the use of 
${\bf A}$ in (\ref{eq:god3.23}) depends on the gauge for ${\bf A}$. 
If for example, ${\tilde{\bf A}}={\bf A}+\nabla\phi$, 
this will induce a change in the function $\lambda$. In 
other words, (\ref{eq:god3.23}) in the new gauge leads to the 
equation $\tilde{\mu}{\tilde{\bf A}}{\bf\cdot} d{\bf x}=d{\tilde\lambda}$. 
If we fix the gauge of ${\bf A}$, then the solution of (\ref{eq:god3.23}) is:
\begin{equation}
{\bf A}=\nu\nabla\lambda, \label{eq:god3.24}
\end{equation}
where $\nu=1/\mu$ and $\mu$ is the integrating factor for the Pfaffian 
equation ${\bf A}{\bf\cdot}d{\bf x}=0$. 
In this case the vectors 
$\boldsymbol{\eta}$, ${\bf B}$ lie in the $\lambda=const.$ surface 
and ${\bf A}$ is normal to the surface (i.e. 
${\bf N}\equiv \hat{\bf A}=\nabla\lambda/|\nabla\lambda|$ is the unit normal to
the surface). The geometrical configuration of ${\bf A}$, ${\bf B}$ and 
$\boldsymbol{\eta}$ and the surface $\lambda=const.$ are 
depicted schematically in Figure \ref{fig:gv-ab-eta}.
\begin{figure}
\hspace*{1cm}
\centering\includegraphics[width=10cm, angle=270]{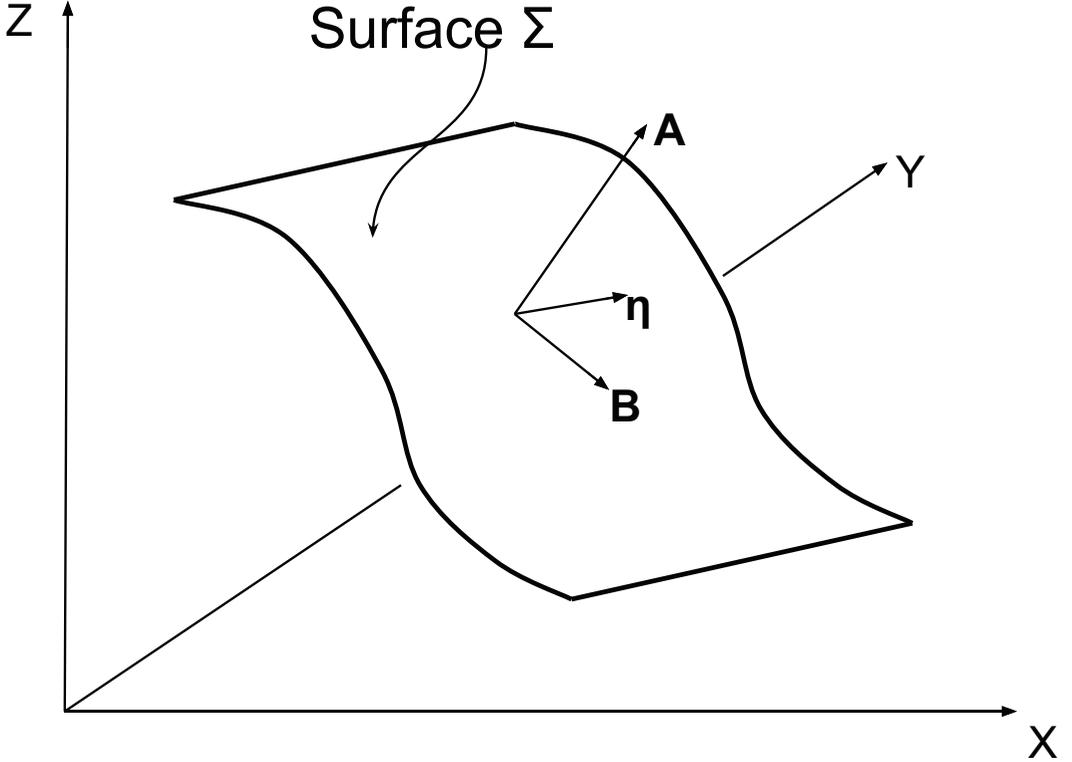}
\caption{The ${\bf B}$ surface $\Sigma$ with normal ${\bf A}$
and Godbillon-Vey field $\boldsymbol{\eta}$, where ${\bf A\cdot B}=0$
and $\boldsymbol{\eta}={\bf A}\times{\bf B}/A^2$.
Both ${\bf B}$ and $\boldsymbol{\eta}$ are in the surface $\Sigma$}.\label{fig:gv-ab-eta}

\end{figure}
\medskip
\medskip
Note that the vectors ${\bf A}$, ${\bf B}$ and $\boldsymbol{\eta}$ are mutually 
orthogonal, and with ${\bf B}$ and $\boldsymbol{\eta}$ lying within the surface 
$\lambda=const.$

The magnetic vector potential ${\bf A}$ can admit a gauge 
potential transformation $\tilde{\bf A}=\bf{A}+\nabla\phi$, i.e.
\begin{equation}
\tilde{\bf A}=\tilde{\nu}\nabla\tilde{\lambda}
={\bf A}+\nabla\phi=\nu\nabla\lambda+\nabla\phi\quad\hbox{where}
\quad {\bf B}=\nabla\times{\bf A}. 
\label{eq:god3.25}
\end{equation} 
In this case  we obtain:
\begin{equation}
\tilde{\bf A}{\bf\cdot B}=\left(\nu\nabla\lambda+\nabla\phi\right){\bf\cdot}
\left(\nabla\nu\times\nabla\lambda\right)=
\nabla\phi{\bf\cdot}(\nabla\nu\times\nabla\lambda)=J
=\frac{\partial(\phi,\nu,\lambda)}{\partial(x,y,z)}. \label{eq:god3.26}
\end{equation}
The zero Jacobian case $J=0$ or ${\bf A\cdot B}=0$  
implies that
\begin{equation}
\phi=\phi(\nu,\lambda)\quad\hbox{for}\quad J=0. \label{eq:god3.27}
\end{equation} 
The gauge transformations (\ref{eq:god3.25}) 
are discussed in Appendix D. 

If $\phi$ has discontinuous jumps 
across some surface, the integral of ${\bf A\cdot B}$ 
over a volume containing the discontinuity surface leads to a non-zero magnetic helicity 
integral over the volume. This implies that there is not a global 
single valued, smooth magnetic vector potential for ${\bf A}$, and that 
a complicated magnetic field topology can arise due to the discontinuity
surface for ${\bf A}$. This possibility is used 
by \cite{Semenov02} to describe the MHD topological soliton
using Euler potentials. 

In the case ${\bf A\cdot B}\neq 0$, 
the  space is not foliated 
into a family of surfaces. One could use  
 Boozer coordinates (\cite{Boozer83, Boozer04}) 
to describe the magnetic field, 
in which case we write:
\begin{equation}
{\bf A}=\nu\nabla\lambda+\psi\nabla\phi\quad\hbox{and}
\quad {\bf B}=\nabla\nu\times\nabla\lambda+\nabla\psi\times\nabla\phi. 
\label{eq:god3.28}
\end{equation}
The Clebsch representations 
(\ref{eq:god3.24})-(\ref{eq:god3.28}) for the Godbillon Vey 
helicity density $H_{gv}$ are discussed 
in Appendix C. 

\subsection{\bf Godbillon-Vey Conservation Law}
 
\begin{proposition}
Using the advected ${\bf A}$ gauge in which the electric field potential $\psi={\bf A\cdot u}$ 
(\cite{Gordin87, Gordin89}), the Godbillon-Vey 
helicity density
$h_{gv}=\boldsymbol{\eta}{\bf\cdot}(\nabla\times \boldsymbol{\eta})$
for MHD flows in which ${\bf A}{\bf\cdot B}=0$ 
satisfies the conservation law:
\begin{equation}
\deriv{h_{gv}}{t}+\nabla{\bf\cdot}\left[ {\bf u} h_{gv}
 +\alpha {\bf B}\right]=0, \label{eq:god3.29}
\end{equation}
where
\begin{align}
\boldsymbol{\eta}=&\frac{{\bf A}\times{\bf B}}{|{\bf A}|^2}, \quad
 h_{gv}=\boldsymbol{\eta}{\bf\cdot}(\nabla\times \boldsymbol{\eta}), \nonumber\\
\alpha=&\frac{2{\bf A\cdot}{\sf\sigma}{\bf\cdot}\boldsymbol{\eta}}{|{\bf A}|^2}
\equiv \frac{{\bf A\cdot}{\sf\sigma}{\bf\cdot}\boldsymbol{\eta}
+\boldsymbol{\eta}{\bf\cdot}{\sf\sigma}{\bf\cdot}{\bf A}}{|{\bf A}|^2}, \nonumber\\
\sf{\sigma}=&\frac{1}{2}\left[\nabla {\bf u}+(\nabla {\bf u})^T
-\frac{2}{3} \left(\nabla{\bf\cdot}{\bf u}\right){\sf I}\right]. 
\label{eq:god3.30}
\end{align}
Here, $\alpha$ depends on the shear tensor of the fluid $\sf{\sigma}$. 
It describes the coupling of the fields ${\bf A}$ and $\boldsymbol{\eta}$ 
due to shear in the flow. For shear-free flows, $\alpha=0$. The 
Godbillon-Vey helicity $H_{gv}=\int_{V_m} h_{gv}d^3x$ for a volume $V_m$ 
moving with the flow is conserved, i.e. $dH_{gv}/dt=0$ where it is assumed
that ${\bf B\cdot n}=0$ on the boundary $\partial V_m$ of $V_m$. 
\end{proposition}

\begin{proof}
The detailed proof follows as a consequence of the analysis of 
\cite{Tur93} and \cite{Webb14a}. A proof is given in appendix A. 
\end{proof}

\section{Godbillon-Vey helicity for ${\bf A\cdot B}\neq 0$}
In this section, we generalize the Godbillon-Vey helicity transport equation in two ways, 
namely (a)\ we determine the form of the transport equation for the case where the 
magnetic helicity $h_m={\bf A\cdot B}\neq 0$ and (b)\ we allow for a general 
electric field potential $\psi$ (i.e we allow for more general gauges for ${\bf A}$, other 
than the advected ${\bf A}$ gauge for which $\psi={\bf A\cdot u}$). 
The underlying idea is that Faraday's equation for ${\bf B}$ 
can be split up into components parallel and perpendicular to 
${\bf A}$ as in (\ref{eq:god3.17}), i.e. ${\bf B}={\bf B}_\parallel+ {\bf B}_\perp$, 
in which ${\bf B}_\parallel=h_m {\bf A}/|{\bf A}|^2$, where $h_m={\bf A\cdot B}$, and 
${\bf B}_\perp$ is related to the Godbillon-Vey field 
$\boldsymbol{\eta}={\bf A}\times{\bf B}/|{\bf A}|^2$ by the 
formula ${\bf B}_\perp=\boldsymbol{\eta}\times {\bf A}$. 

\begin{proposition}\label{prop4.1}
The transport equation for the Godbillon-Vey helicity 
$h_{gv}=\boldsymbol{\eta}{\bf\cdot}(\nabla\times\boldsymbol{\eta})$ where 
$\boldsymbol{\eta}={\bf A}\times{\bf B}/|{\bf A}|^2$, for the general case where 
$h_m={\bf A}{\bf\cdot}{\bf B}\neq 0$, and for a general gauge for ${\bf A}$ 
has the form:
\begin{equation}
\deriv{h_{gv}}{t}+\nabla{\bf\cdot}(h_{gv} {\bf u})=Q, \label{eq:god4.1}
\end{equation}
where the source term $Q$ in (\ref{eq:god4.1}) is given by:
\begin{align}
Q=&{\bf S}{\bf\cdot}\nabla\times\boldsymbol{\eta}
+\boldsymbol{\eta}{\bf\cdot}\nabla\times{\bf S}, \nonumber\\
{\bf S}=&\frac{h_m}{|{\bf A}|^4} {\bf A}\times\left(2{\sf\sigma}{\bf\cdot}{\bf A} 
+\nabla\zeta\right)
+\frac{{\bf A}{\bf\cdot}\nabla\zeta}{|{\bf A}|^2} \boldsymbol{\eta} 
+\alpha {\bf A}, \nonumber\\
\alpha=&\frac{\left(2 {\bf A}{\bf\cdot}{\sf\sigma}{\bf\cdot}\boldsymbol{\eta}
+\boldsymbol{\eta}{\bf\cdot}\nabla\zeta\right)}
{|{\bf A}|^2}, 
\quad \zeta=\psi-{\bf A}{\bf\cdot}{\bf u}, \label{eq:god4.2}
\end{align}
and ${\sf\sigma}$ is the fluid velocity shear tensor in (\ref{eq:god3.30}).
 Here, ${\bf S}$ 
is the source term in the Godbillon-Vey field evolution equation: 
\begin{equation}
\boldsymbol{\eta}_t-{\bf u}\times(\nabla\times\boldsymbol{\eta}) 
+\nabla({\bf u}{\bf\cdot}\boldsymbol{\eta})={\bf S}. \label{eq:god4.3}
\end{equation}
In the special case ${\bf A}{\bf\cdot}{\bf B}=0$ and $\zeta=0$ ($\psi={\bf A}{\bf\cdot}{\bf u}$)
the Godbillon-Vey transport equation (\ref{eq:god4.1}) reduces to the 
conservation law (\ref{eq:god3.29}). In the advected ${\bf A}$ gauge 
($\zeta=0$ and $\psi={\bf A\cdot u}$), (\ref{eq:god4.2}) gives the simplified formulae:
\begin{equation}
{\bf S}={\bf S}^0=\frac{2 h_m}{|{\bf A}|^4} {\bf A}\times({\sf\sigma}{\bf\cdot}{\bf A})
+\alpha {\bf A}, 
\quad\alpha\equiv\alpha_0=\frac{2({\bf A}{\bf\cdot}{\sf\sigma}{\bf\cdot}\boldsymbol{\eta})}
{|{\bf A}|^2}. \label{eq:god4.4}
\end{equation}
\end{proposition}

\begin{proof}
The proof is given in Appendix B. 
\end{proof}

\section{The \cite{Low90} Force Free Magnetic Fields}
In this section we investigate the Godbillon-Vey helicity of the \cite{Low90}
force free magnetic fields. Both \cite{Low90} and \cite{Prasad14} 
used these fields to discuss solar photospheric magnetic fields. 
\cite{Prasad14} investigated models of force free magnetic fields in order  
to describe solar magnetic fields  observed by the Hinode spectro-polarimeter.
They studied the relative magnetic helicity and magnetic free energy of 
magnetically active regions (AR's) on the Sun, both before and after 
solar flares. 

The force-free magnetic fields arise in low beta magnetic fields 
in a highly conducting plasma when the dominant force in the 
magneto-static force balance is the ${\bf J\times B}$ force. 
In this case, the approximate force balance equation reduces to 
${\bf J\times B}=(\nabla\times{\bf B})\times{\bf B}/\mu_0=0$. For force-free
 fields, ${\bf B}$ is to lowest order parallel
to the current ${\bf J}$ so that: 
\begin{equation}
\nabla\times{\bf B}=\alpha {\bf B}, \quad \nabla{\bf\cdot}{\bf B}=0, 
\quad {\bf B}{\bf\cdot}\nabla\alpha=0. \label{eq:ll1}
\end{equation}
For linear force free fields $\alpha$ is taken to 
be constant (e.g. \cite{Chandrasekhar56}). For nonlinear force free  
fields, the function $\alpha$ can be a nonlinear function of the magnetic 
vector potential ${\bf A}$, or of a component of ${\bf A}$. 
From (\ref{eq:ll1})
we obtain the condition:
\begin{equation}
\nabla{\bf\cdot}(\nabla\times{\bf B})=\nabla{\bf\cdot}(\alpha {\bf B})
={\bf B}{\cdot}\nabla\alpha=0. \label{eq:ll2}
\end{equation}
Thus $\alpha$ is constant along a field line.  

The nonlinear force-free magnetic fields of \cite{Low90} have the form:
\begin{equation}
{\bf B}=\nabla \psi\times\nabla\phi+\frac{Q{\bf e}_\phi}{r\sin\theta}
\equiv \frac{\nabla\psi\times{\bf e}_\phi+Q{\bf e}_\phi}{r\sin\theta},  
\label{eq:ll3}
\end{equation}
where $(r,\theta,\phi)$ are spherical polar coordinates, and
\begin{equation}
{\bf e}_r=\deriv{\bf r}{r},
\quad {\bf e}_\theta=\frac{1}{r}\deriv{\bf r}{\theta},
\quad {\bf e}_\phi=\frac{1}{r\sin\theta}\deriv{\bf r}{\phi}, \label{eq:ll4}
\end{equation}
are orthonormal unit vectors in the $r$, $\theta$, and $\phi$ directions. 
In (\ref{eq:ll3})  $A=A(r,\theta)$ and $Q=Q(r,\theta)$. 

Using the magnetic field representation (\ref{eq:ll3}), 
(\ref{eq:ll1}) give the equations:
\begin{align}
\frac{1}{r^2\sin\theta} 
\left(\deriv{Q}{\theta}-\alpha \deriv{\psi}{\theta}\right)=&0, \label{eq:ll5}\\
\frac{1}{r\sin\theta}\left(-\deriv{Q}{r}+\alpha \deriv{\psi}{r}\right)=&0, 
\label{eq:ll6}\\
\frac{1}{r\sin\theta}\frac{\partial^2 \psi}{\partial r^2} 
+\frac{1}{r^3}\derv{\theta}
\left(\frac{1}{\sin\theta}\deriv{\psi}{\theta}\right)
+\frac{\alpha Q}{r\sin\theta}=&0, \label{eq:ll7}
\end{align}
as the components of the force balance equation in the $r$, 
$\theta$ and $\phi$ directions respectively. 
 
From (\ref{eq:ll5}) and (\ref{eq:ll6}) the compatibility condition:
\begin{equation}
\deriv{\psi}{r}\left(\deriv{Q}{\theta}-\alpha\deriv{\psi}{\theta}\right)
-\deriv{\psi}{\theta}\left(\deriv{Q}{r}-\alpha \deriv{\psi}{r}\right)
\equiv \frac{\partial(Q,\psi)}{\partial (\theta,r)}=0, \label{eq:ll8}
\end{equation}
implies $Q=Q(\psi)$. Similarly, the condition ${\bf B\cdot}\nabla\alpha=0$ 
in (\ref{eq:ll1}) requires:
\begin{equation}
B_r\deriv{\alpha}{r}+\frac{B_{\theta}}{r}\deriv{\alpha}{\theta}
\equiv\frac{1}{r^2\sin\theta} 
\frac{\partial(\psi,\alpha)}{\partial(\theta,r)}=0.
\quad \label{eq:ll9}
\end{equation}
Equation(\ref{eq:ll9}) requires that $\alpha=\alpha(\psi)$. 
(\ref{eq:ll5}) 
and (\ref{eq:ll6}) 
gives the equations:
\begin{equation}
\deriv{\psi}{\theta}\left(\frac{dQ}{d\psi}-\alpha\right)=0, \quad \deriv{\psi}{r}
\left(-\frac{dQ}{d\psi}+ \alpha\right)=0, 
\quad \hbox{and}\quad \alpha=\frac{dQ}{d\psi}.
\label{eq:ll9a}
\end{equation}
Using $\alpha=dQ/d\psi$ and $Q=Q(\psi)$ in (\ref{eq:ll7}) gives the equation:
\begin{equation}
\frac{\partial^2 \psi}{\partial r^2}+\frac{(1-\mu^2)}{r^2} 
\frac{\partial^2 \psi}{\partial\mu^2} +\frac{dQ(\psi)}{d\psi} Q(\psi)=0, 
\label{eq:ll10}
\end{equation}
where $\mu=\cos\theta$. 
Thus, the nonlinear force-free magnetic field equation (\ref{eq:ll1}) reduces 
to the second order partial differential equation (\ref{eq:ll10})  for 
$\psi$, with ignorable spherical polar coordinate $\phi$. 
Equation (\ref{eq:ll10}) is analogous to the Grad-Shafranov 
equation for MHD equilibria with ignorable coordinate $\phi$. 

From \cite{Low90} and \cite{Prasad14} (\ref{eq:ll10}) admits 
separable solutions for $\psi$ of the form:
\begin{equation}
\psi=\frac{P(\mu)}{r^n},\quad Q=a \psi^{1+1/n}, \label{eq:ll11}
\end{equation}
where $a$ is a constant and $P(\mu)$ satisfies the nonlinear second order 
differential equation:
\begin{equation}
(1-\mu^2) \frac{d^2P}{d\mu^2}+n(n+1) P+ a^2 
\left(\frac{n+1}{n}\right) P^{1+2/n}=0. \label{eq:ll12}
\end{equation}
For $a=0$, the solution of (\ref{eq:ll12}), which is regular at $\mu=\pm1$ is:
\begin{equation}
P(\mu)=(1-\mu^2)^{1/2} P_n^1(\mu), \label{eq:ll13}
\end{equation}
where
$P_n^m(\mu)$ is Legendre's associated function.

Here $a=0$ implies $Q=0$ and $\alpha = 0$, which represent the potential field (untwisted) 
solutions. For the non-potential cases ($a \neq 0$), the above nonlinear equation has to be 
solved numerically as an eigenvalue problem (for different values of $a$) subject to the 
boundary conditions $P=0$ at $\mu = \pm 1$. The solutions of (\ref{eq:ll12}) for the restrictive 
case of $n=1$ were presented in \cite{Low90}, which were constrained to due to an inherent 
singularity in $P$ at $\mu = 0$ for higher values of $n$. These solutions were later extended 
for higher odd values of $n$ by \cite{Prasad14} through the transformation 
$P(\mu) = \sqrt(1-\mu^2)F(\mu)$ and then solving the \ref{eq:ll12}) in terms of $F$ where
\begin{equation}
(1-\mu^2)\frac{d^2F}{d\mu^2} - 2\mu \frac{dF}{d\mu}+\left[n(n+1)-\frac{1}{(1-\mu^2)}\right]F(\mu)
+a^2\frac{(n+1)}{n}F^{(n+2)/n}(1-\mu^2)^{1/n} = 0. \label{eq:ll13a}
\end{equation}
with the boundary conditions $F(\mu) = 0$ at $\mu=\pm 1$.

Following \cite{Prasad14} we search for a two-dimensional magnetic vector 
potential of the form:
\begin{equation}
{\bf A}=A_{\theta} {\bf e}_{\theta}+A_\phi {\bf e}_{\phi}, \label{eq:ll14}
\end{equation}
where ${\bf B}=\nabla\times{\bf A}$.  Using (\ref{eq:ll3}), 
(\ref{eq:ll11}), (\ref{eq:ll12}), and (\ref{eq:ll14}) we obtain the equations:
\begin{align}
B_r=&\frac{1}{r\sin\theta}\left[\derv{\theta}(A_\phi \sin\theta) 
-\deriv{A_\theta}{\phi}\right]=-\frac{dP/d\mu}{r^{n+2}}, \nonumber\\
B_\theta=&-\frac{1}{r}\derv{r}(r A_\phi)= \frac{n P(\mu)}{r^{n+2}\sin\theta} 
, \nonumber\\
B_\phi=&\frac{1}{r}\derv{r}(rA_\theta)= \frac{a [P(\mu)]^{(n+1)/n}}
{r^{n+2}\sin\theta}. \label{eq:ll15}
\end{align}
Integration of (\ref{eq:ll15}) gives the solutions:
\begin{equation}
A_\theta=-\frac{1}{n} \frac{a[P(\mu)]^{(n+1)/n}} {r^{n+1}\sin\theta},
\quad A_\phi=\frac{P(\mu)}{r^{n+1}\sin\theta}. \label{eq:ll16}
\end{equation}
for $A_\theta$ and $A_\phi$. 

Using (\ref{eq:ll15}) and (\ref{eq:ll16}) it follows that:
\begin{equation}
{\bf A\cdot B}=A_\theta B_\theta+A_\phi B_\phi=0. \label{eq:ll17}
\end{equation}

Thus ${\bf A}{\bf\cdot}\nabla\times{\bf A}=0$,  which implies that 
the Pfaffian ${\bf A}{\bf\cdot}d{\bf x}=0$ admits an integrating factor,
$\mu$ such that $\mu {\bf A}{\bf\cdot}d{\bf x}=d\Phi$ 
where $\Phi({\bf x})=const.$ is a foliation with normal ${\bf A}$. 
The magnetic field  ${\bf B}$ and the Godbillon-Vey 
field $\boldsymbol{\eta}={\bf A}\times{\bf B}/|{\bf A}|^2$ lie on the foliation.
Using (\ref{eq:ll15})-(\ref{eq:ll16}), the components of 
$\boldsymbol{\eta}$ are given by the equations:
\begin{equation}
\eta_r=-\frac{n}{r},\quad 
\eta_\theta=-\frac{(dP/d\mu)\sin\theta}{rP[1+(a^2/n^2)P^{2/n}]}, 
\quad 
\eta_\phi=-\frac{a P^{(1/n-1)} (dP/d\mu)\sin\theta} 
{n r P[1+(a^2/n^2)P^{2/n}]}. \label{eq:ll18}
\end{equation}
In the evaluation of (\ref{eq:ll18}) we used the formula:
\begin{equation}
|{\bf A}|^2=\frac{[1+(a^2/n^2)P^{2/n}] P^2}{r^{2n+2}\sin^2\theta}. 
\label{eq:ll19}
\end{equation}
Calculating $\nabla\times\boldsymbol{\eta}$ we obtain:
\begin{equation}
(\nabla\times\boldsymbol{\eta})_r=\frac{a}{nr^2}
\frac{d}{d\mu}\left(\frac{P^{(1/n-1)} (dP/d\mu) (1-\mu^2)}
{[1+(a^2/n^2) P^{2/n}]}\right), 
\quad (\nabla\times\boldsymbol{\eta})_\theta
=(\nabla\times\boldsymbol{\eta})_\phi=0. \label{eq:ll20}
\end{equation}

Using (\ref{eq:ll18})-(\ref{eq:ll20}) we obtain the Godbillon-Vey helicity 
density $h_\eta$ as:
\begin{equation}
h_\eta=\boldsymbol{\eta}{\bf\cdot}\nabla\times\boldsymbol{\eta}
=-\frac{a}{r^3} \frac{d}{d\mu} 
\left(\frac{P^{(1/n-1)} dP/d\mu (1-\mu^2)}
{[1+(a^2/n^2) P^{2/n}]}\right). \label{eq:ll21}
\end{equation}
Note that for the potential field with $a=0$, the Godbillon-Vey 
helicity density $h_{gv}\equiv h_{\eta}$ is zero. Only for $a\neq 0$
is there a non-zero Godbillon Vey helicity.  
The parameter $a$ is an eigen-value in the nonlinear 
force free fields obtained by \cite{Low90} and \cite{Prasad14}. 
The eigenfunctions $P(\mu)$ are characterized by the label $n$ 
(see \ref{eq:ll12}), which determines the radial dependence of the solution for
$\psi$. A similar classification applies to the $F(\mu)$, where $P(\mu)=\sqrt{(1-\mu^2)}$. The eigen-functions are also labelled by the index $m$ 
where $m=1$ corresponds to the lowest possible value of the eigenvalue 
$a$ that fits the boundary conditions $F(\mu)=0$ at $\mu=\pm 1$. 

\begin{figure}
\centering\includegraphics[width=12cm, angle=0]{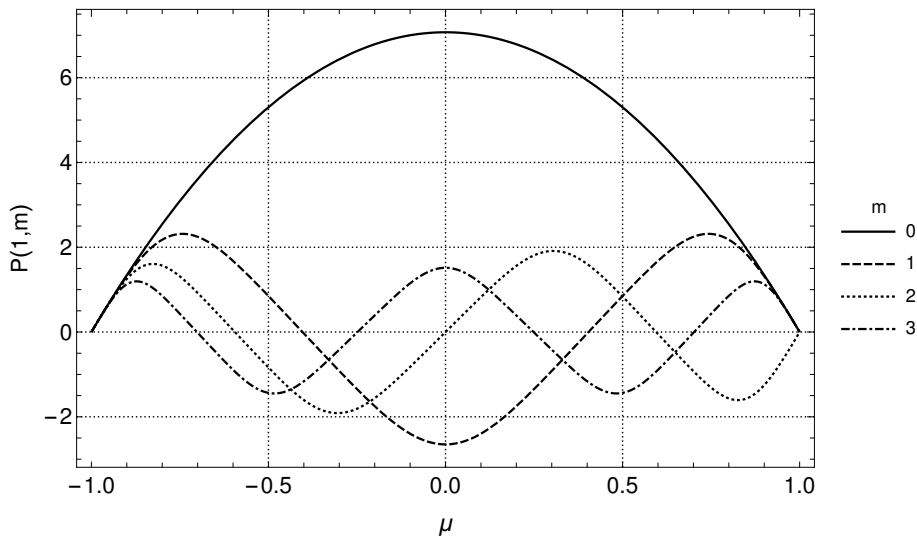}
\caption{\cite{Low90} force free magnetic field eigen-functions $P(n,m)$ 
for $n=1$ versus $\mu$ where $\mu=\cos\theta$ and $(r,\theta,\phi)$ 
are spherical polar coordinates. $m=0,1,2,3$.}  
\label{fig:p1_demo2a}
\end{figure}
Figure 4 shows the eigen-functions $P(n,m;\mu)$ versus 
$\mu$ for $-1<\mu<1$, for $n=1$ and  
$m=0,1,2,3$. The case $m=0$ corresponds to the potential field case
where $P(\mu)=1-\mu^2$. For $m=1$ and $m=3$ $P(\mu)$ is even in $\mu$,
 but for $m=2$ $P(\mu)$  is odd in $\mu$ (see also \cite{Low90} and
\cite{Prasad14}).

\begin{figure}
\centering\includegraphics[width=8cm, angle=0]{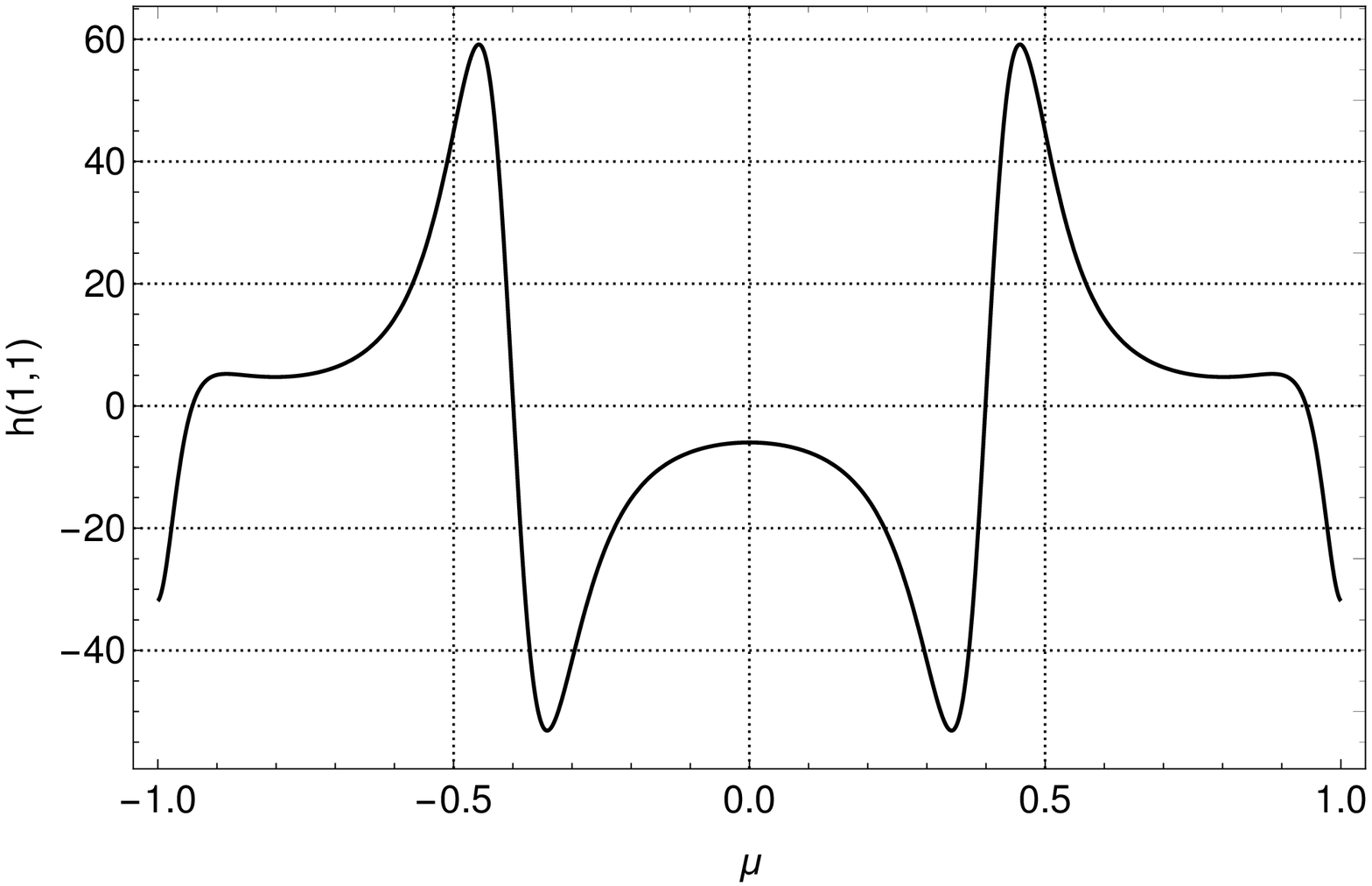}\\
\centering\includegraphics[width=8cm, angle=0]{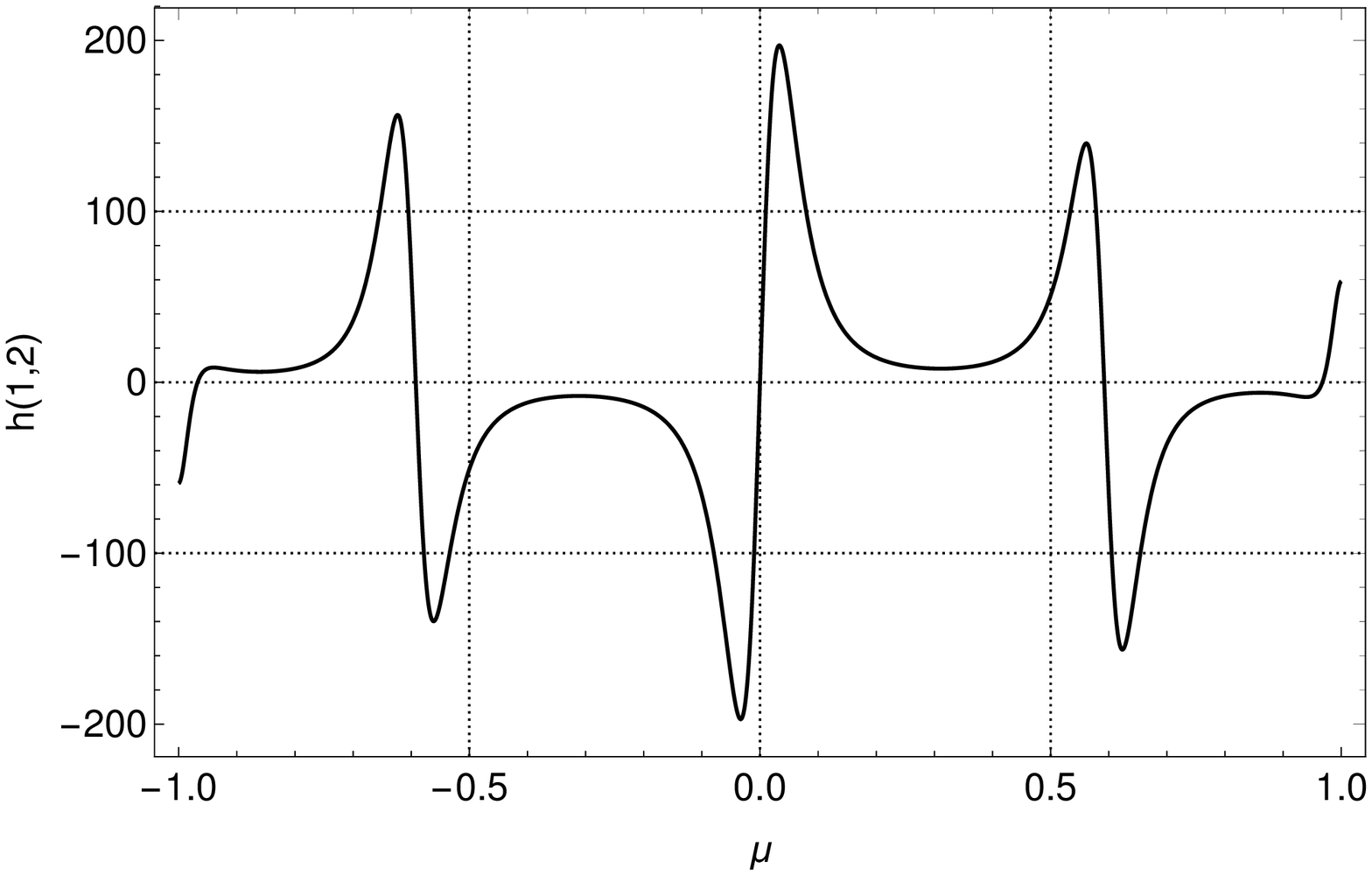}\\
\centering\includegraphics[width=8cm, angle=0]{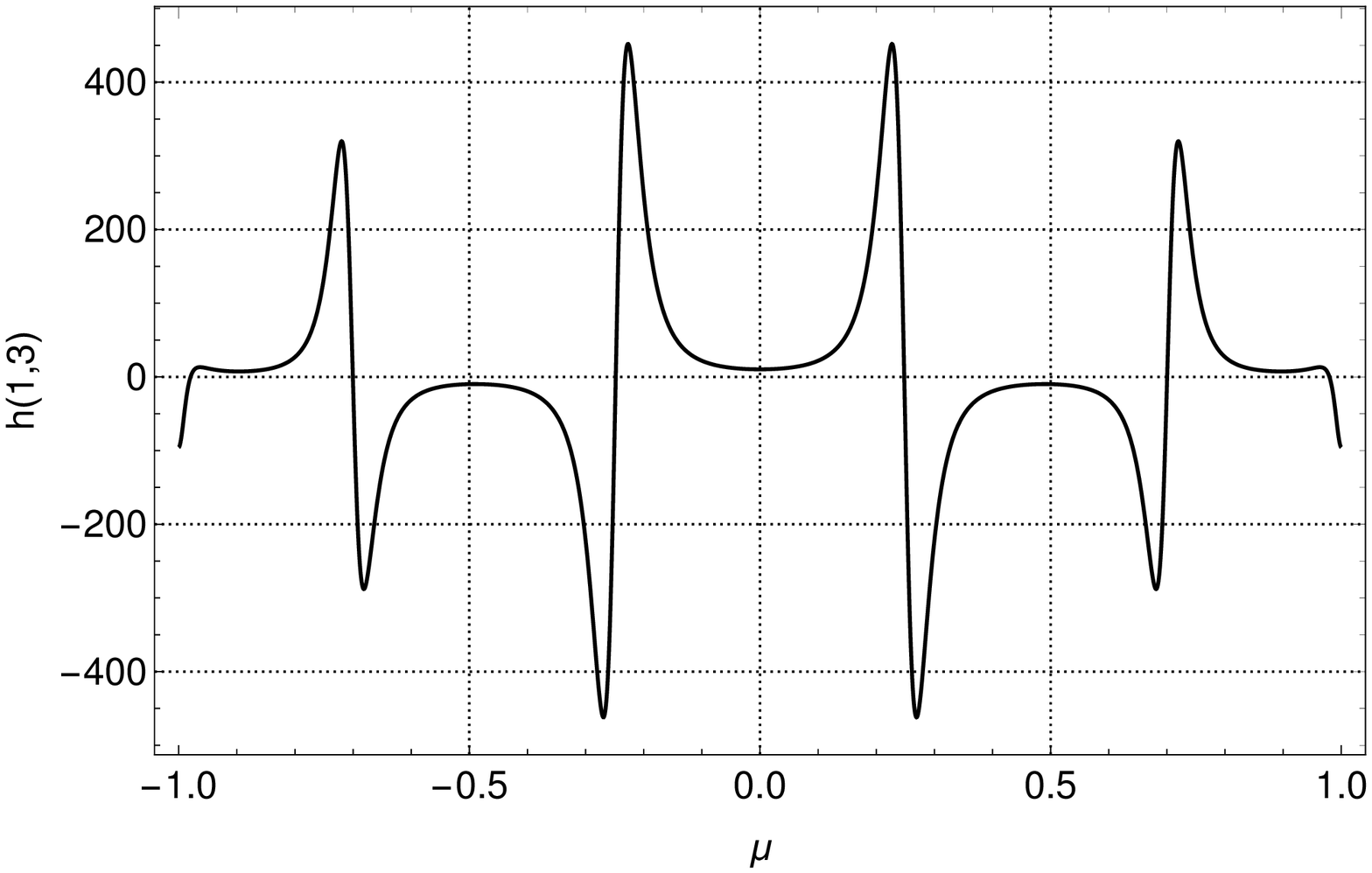}
\caption{Godbillon-Vey helicity density $h(n,m)$ versus $\mu$ for the 
\cite{Low90} force free magnetic field (\ref{eq:ll15}). 
$h(n,m)\equiv h_\eta$ is given by (\ref{eq:ll21}). The parameter $n=1$ 
and $m=1,2,3$ (top to bottom). Note $h(1,1)$ and $h(1,3)$ are even in $\mu$
but $h(1,2)$ (middle panel) is odd in $\mu$.}
\label{fig:h-godbillon}
\end{figure}
Figure 5 shows the Godbillon-Vey helicity $h\equiv h_\eta$ 
(equation (\ref{eq:ll21}) versus 
$\mu$ ($-1\le \mu\le 1$) for $r=1$,  
 for $n=1$ and for $m=1,2,3$. 
The  panels corresponds to $m=1$ (top panel) $m=2$ (middle panel) and $m=3$
(bottom panel). The helicity densities $h(1,m)$ versus $\mu$ 
are shown for $m=1,2,3$ from top to bottom. For $m=1$ and $m=3$ the  helicity 
densities $h(1,m)$ are even in $\mu$ but $h(1,2)$ for $m=2$ is odd in $\mu$.  
Note the existence of positive and negative values of $h\equiv h_{gv}$ 
as a function of $\mu$. The maximum and minimum values of $h$ 
increases with $m$. 

\begin{figure}
\centering\includegraphics[width=8cm, angle=0]{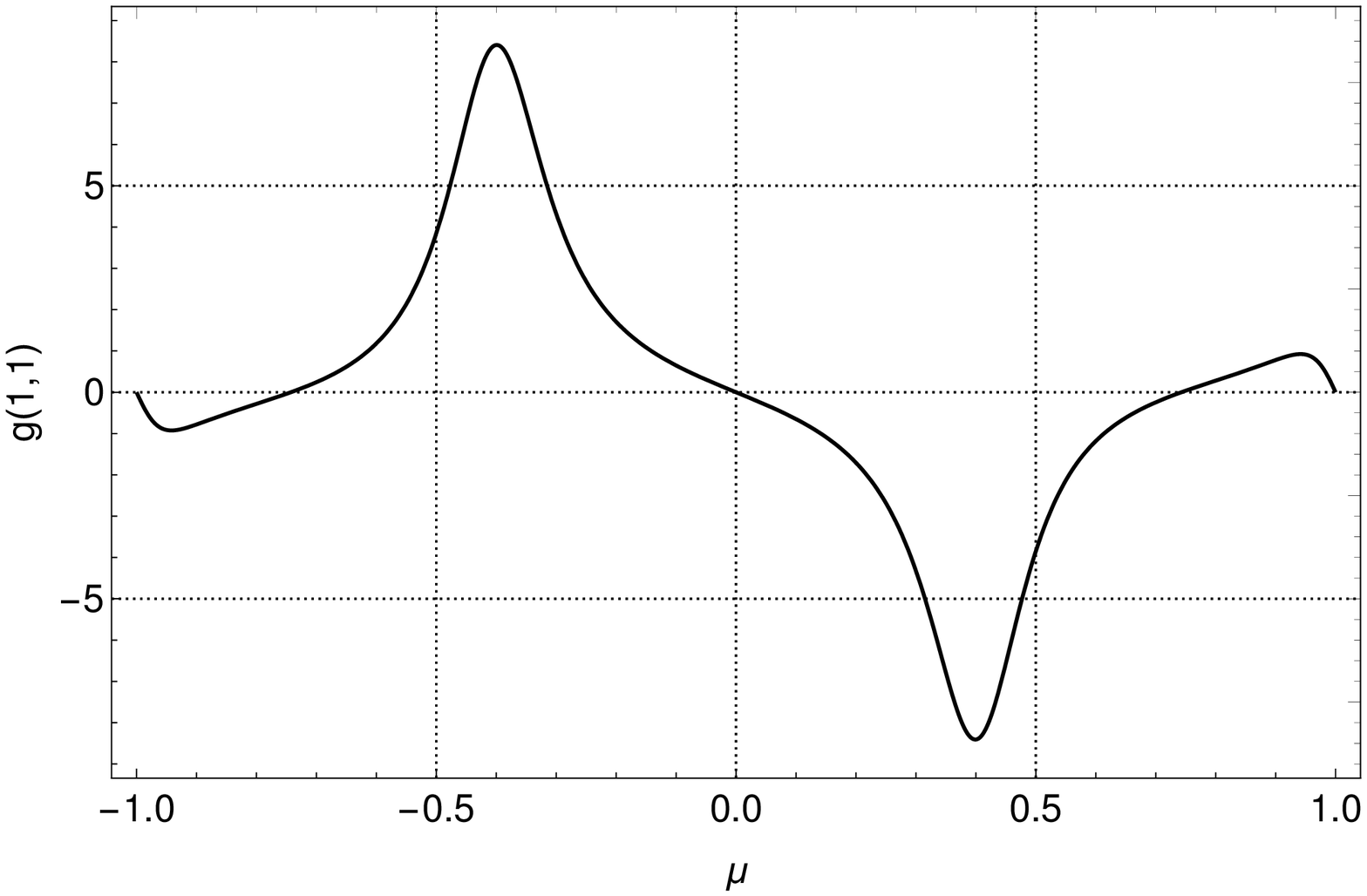}\\
\centering\includegraphics[width=8cm, angle=0]{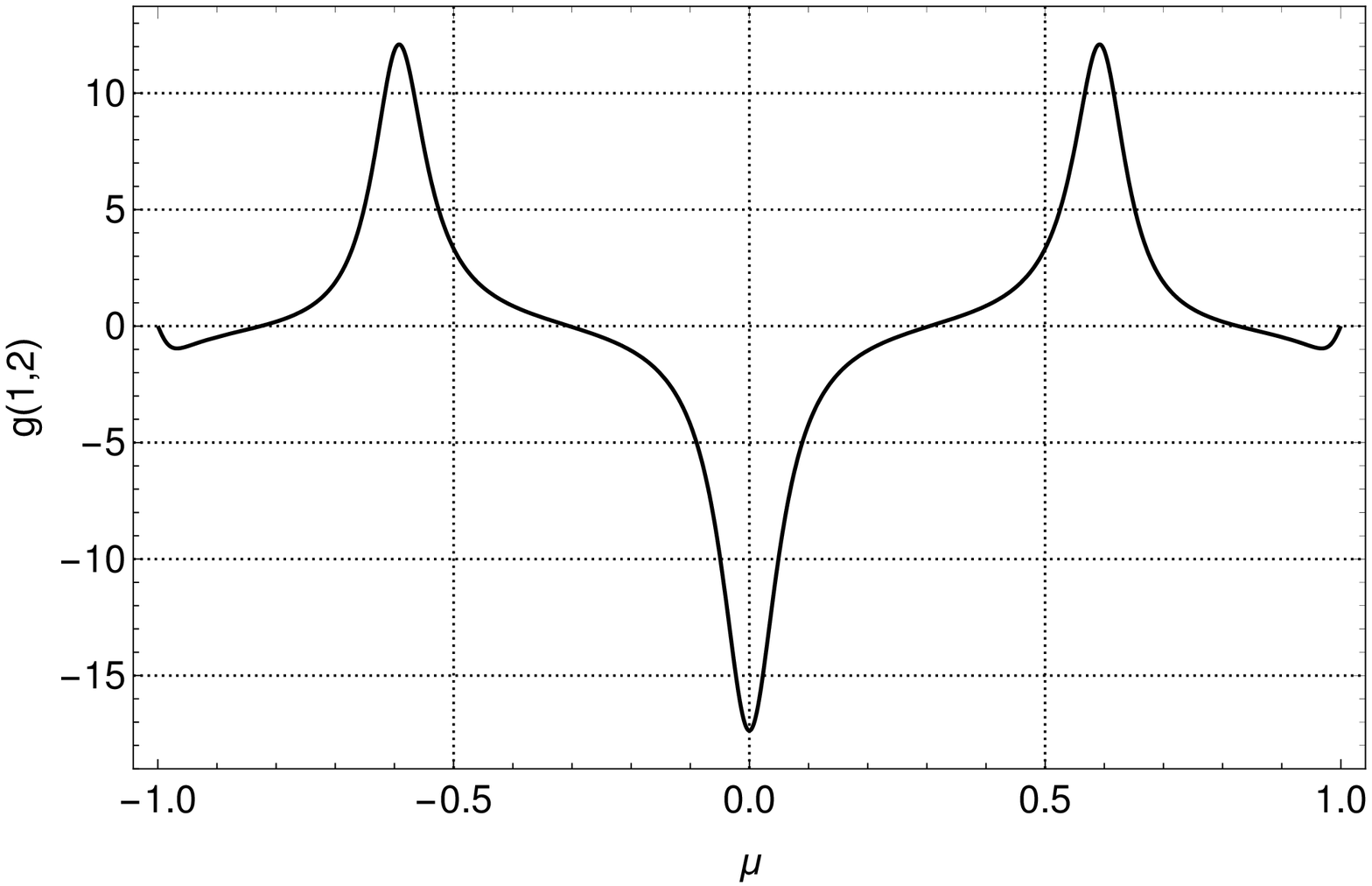}\\
\centering\includegraphics[width=8cm, angle=0]{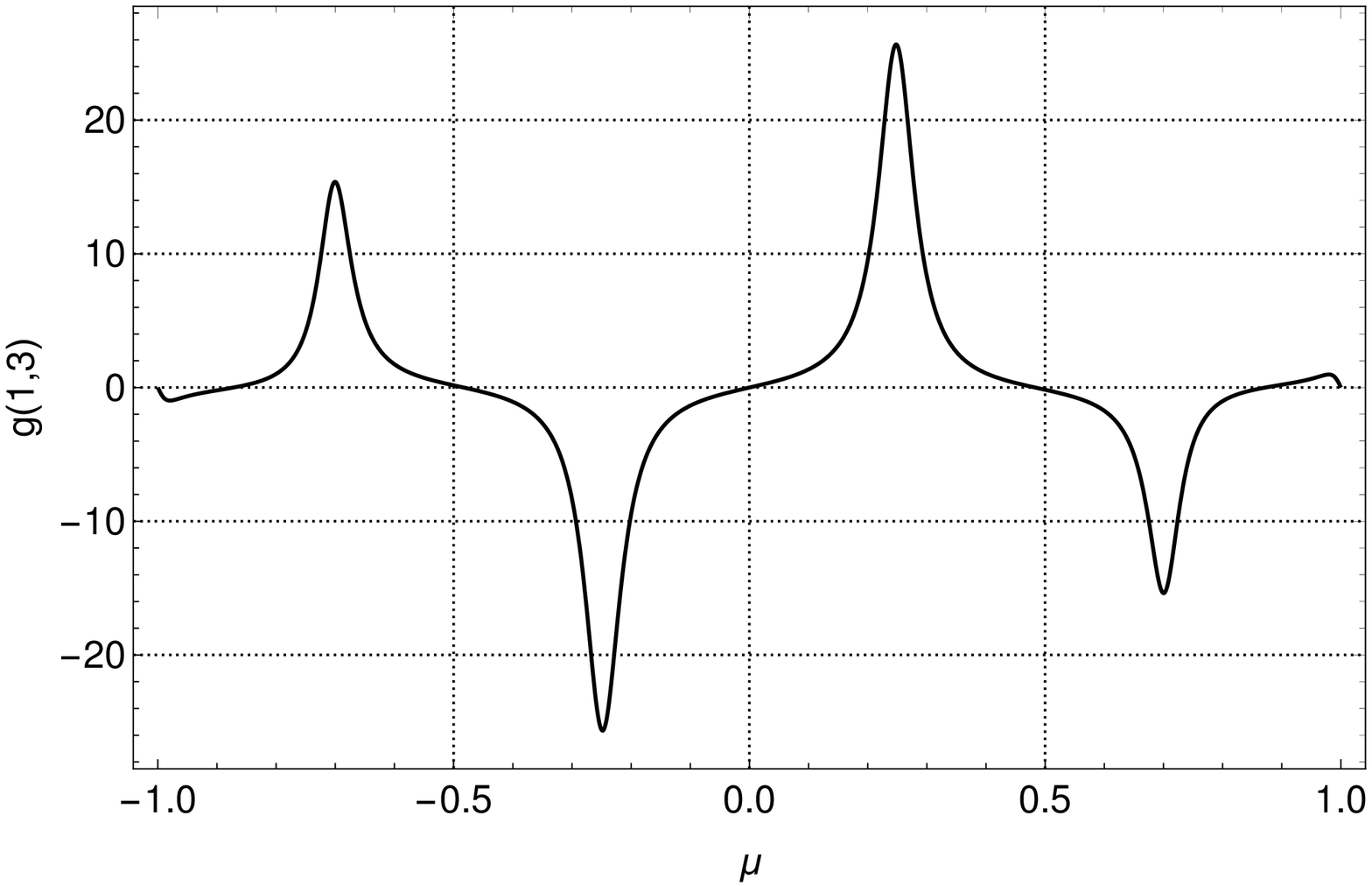}\\
\caption{Cumulative Godbillon-Vey helicity $g(n,m)$ versus $\mu$ for the
\cite{Low90} force free magnetic field (\ref{eq:ll15}).
$g(n,m)\equiv g(\mu)$ is given by (\ref{eq:ll21a}). The parameter $n=1$
and $m=1,2,3$ (top to bottom). Note $g(1,1)$ and $g(1,3)$ are odd in $\mu$
but $g(1,2)$ (middle panel) is even in $\mu$.}
\label{fig:g-godbillon}
\end{figure}
Perhaps of more physical interest is the net Godbillon-Vey 
helicity integral for $r=1$, namely:
\begin{equation}
g(\mu)=\int_{\mu=-1}^\mu d\mu' h_\eta(\mu')= -a\left[
\left(\frac{P^{(1/n-1)}(\mu') dP(\mu')/d\mu' (1-\mu^{'2})}
{[1+(a^2/n^2) P(\mu')^{2/n}]}\right)\right]_{\mu'=-1}^{\mu'=\mu}.
 \label{eq:ll21a}
\end{equation}
The plots of $g(\mu)$ in Figure 6 show that $g(\mu)=0$ at $\mu=1$, 
i.e. the net Godbillon Vey helicity 
integral is zero. Note that $g(\mu)$ is odd in $\mu$ for $m=1$ and $m=3$
but is even in $\mu$ for $m=2$. 

In general $h_\eta\neq 0$, but its integral over the spherical shell
$r_1<r<r_2$ is zero. The integral of $h_\eta$ over the northern 
hemisphere $0<\theta<\pi/2$ is minus that over the southern hemisphere 
$-\pi/2<\theta<0$.  

\subsection{\bf Clebsch Potential Representations}

Because ${\bf A}{\bf\cdot}\nabla\times{\bf A}=0$, it follows
that the Pfaffian ${\bf A}{\bf\cdot}d{\bf x}=0$ admits an integrating factor
$\mu$, such that $\mu {\bf A}=\nabla\Phi$ where $\Phi=const.$ is a foliation.
This means that ${\bf A}$ and 
${\bf B}$ have the Clebsch representation:
\begin{equation}
{\bf A}=\chi\nabla\Phi,\quad {\bf B}=\nabla\chi\times\nabla\Phi, 
\label{eq:ll22}
\end{equation}
where
\begin{equation}
\Phi=\int^\xi \frac{d\xi'}{G(\xi')}, \quad  
\xi=\phi+\frac{a}{n} \int^\mu \frac{P(\mu')^{1/n}}{(1-\mu^{'2})}\ d\mu', 
\quad \chi=P(\mu)\frac{G(\xi)}{r^n}. \label{eq:ll23}
\end{equation}
A derivation of the formulas
(\ref{eq:ll22})-(\ref{eq:ll23}) are given in Appendix F.
Note that   
 ${\bf B}$
 is independent of the choice of the arbitrary 
function $G(\xi)$ in the Clebsch representation.   
 Also note that $A_r=0$ for the 
solution (\ref{eq:ll22})-(\ref{eq:ll23}). 
Both ${\bf B}$ and 
$\boldsymbol{\eta}$ lie on the foliation $\Phi=const.$.  
The magnetic field lines are located on the intersection of the 
$\Phi=const.$ surfaces and the $\chi=const.$ surfaces. Note that: 
\begin{equation}
{\bf A}=\chi\nabla\Phi= \nabla(\chi\Phi)-\Phi\nabla\chi\equiv -\Phi\nabla\chi, 
\label{eq:ll24}
\end{equation}
In this latter representation, ${\bf A}$ is normal 
to the $\chi=const.$ surface and $\psi=\chi\Phi$ 
is a gauge potential.  
It is clear that the latter form of ${\bf A}$ in (\ref{eq:ll24}) 
is also a valid repesentation for ${\bf A}$, that gives the \cite{Low90}
nonlinear force free magnetic field ${\bf B}$ (\ref{eq:ll15}) 
in the form (\ref{eq:ll22}). 

The simplest form for $G(\xi)$ in (\ref{eq:ll23}) is $G(\xi)=1$. In this case
\begin{equation}
\Phi=\xi\quad\hbox{and}\quad \chi=\frac{P(\mu)}{r^n}. \label{eq:ll26}
\end{equation}
which gives:
\begin{equation}
{\bf A}=\chi\nabla\xi=\frac{P(\mu)}{r^{n+1}\sin\theta}\left[{\bf e}_\phi-\frac{a}{n}P(\mu)^{1/n} {\bf e}_\theta\right], \label{eq:ll27}
\end{equation}
which is the \cite{Prasad14} form of ${\bf A}$ 
for the \cite{Low90}  field. 

From (\ref{eq:ll24}) the magnetic vector potential:
\begin{equation}
{\bf A}^{(2)}=-\xi\nabla\chi
=\frac{\xi}{r^{n+1}}
\left[ nP(\mu) {\bf e}_r +\frac{dP}{d\mu} 
\sin\theta {\bf e}_\theta\right], \label{eq:ll28}
\end{equation}
 also gives rise to the \cite{Low90} 
force free field. 

From (\ref{eq:ll28}) it follows that ${\bf A}^{(2)}$ is normal to 
the $\chi=const.$
foliation:
\begin{equation}
r=\left[\frac{P(\mu)}{\chi}\right]^{1/n},
\label{eq:ll29}
\end{equation}
where
\begin{equation}
{\bf x}=r\left(\sin\theta\cos\phi,\sin\theta\sin\phi,
\cos\theta\right). \label{eq:ll30}
\end{equation}

\begin{figure}
\centering\includegraphics[width=12cm, angle=0]{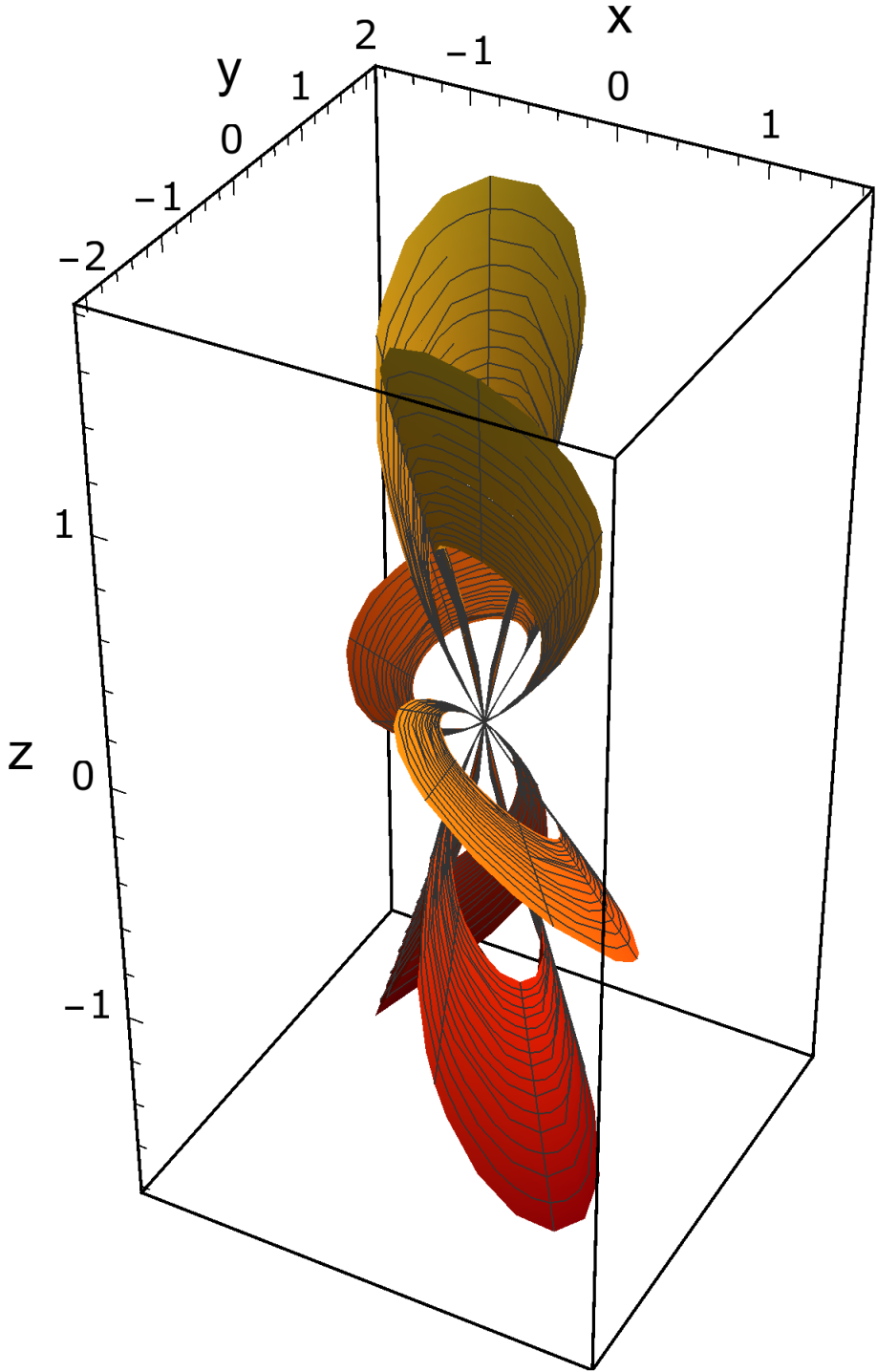}
\caption{The surface $\xi=1$ for the \cite{Low90} solution, where 
$\xi$ and $\chi$ are the Euler potentials 
defining ${\bf A}$ in (\ref{eq:ll27}). 
The surface is described by (\ref{eq:ll30a}) et seq. 
where ${\bf x}={\bf X}(\mu,\chi)$. The parameters $n=m=1$.}
\label{fig:xi-1-surface}
\end{figure}

To plot the foliation $\xi=const.$, with normal ${\bf n}=\nabla\xi/|
\nabla\xi|$, note that 
\begin{equation}
{\bf x}=\bf{X}(\theta,\chi)=r\left(\sin\theta\cos\phi,\sin\theta\sin\phi,
\cos\theta\right). \label{eq:ll30a}
\end{equation}
is a parametric form of the surface, where:
\begin{align}
\phi=&\xi-\int^\mu\frac{\gamma(\mu)}{1-\mu^2}\ d\mu,\quad 
 \gamma(\mu)=\frac{a}{n}
P(\mu)^{1/n}, \nonumber\\
r=&\left(\frac{P(\mu)}{\chi}\right)^{1/n}, \quad \mu=\cos\theta. 
\label{eq:ll30b}
\end{align}
Thus, the $\xi=const.$ surface can be described by the two independent 
parameters $(\theta,\chi)$, where $r=r(\mu,\chi)$ and $\phi=\phi(\xi,\mu)$ 
are given by (\ref{eq:ll30b}). The parametric representation of 
the $\xi=const.$ surface is a standard approach in differential geometry
(e.g. \cite{Lipschutz69}), from which one can extract the metric, or first fundamental form $I$:
\begin{equation}
I=g_{11}(dq^1)^2+2 g_{12}dq^1 dq^2+g_{22}(dq^2)^2, \label{eq:ll30c}
\end{equation}
where
\begin{equation}
q^1=\theta,\quad q^2=\chi,\quad g_{ij}={\bf x}_{q^i}{\bf \cdot}{\bf x}_{q^j}, 
\quad i,j=1,2, \label{eq:ll30d}
\end{equation}
is the metric for the surface. 

Figure 7 shows the surface $\xi=1$ generated by varying $\mu=\cos\theta$, 
($-1<\mu<1$), and by varying the parameter $\chi$ in the range 
$1<\chi<2$, where ${\bf A}=\chi\nabla\xi$ in (\ref{eq:ll27}). The surface
from (\ref{eq:ll30a}) has the form ${\bf x}={\bf X}(\mu,\chi)$. The parameters
$n=m=1$. The surface apparently several branches  which all pass through the 
origin and which fan out at larger $(x,y,z)$. The magnetic  field in fact 
diverges as $r\to 0$ at the origin. This means that for a realistic 
field, it is necessary to exclude the origin (e.g. limit the field 
to a region $r>r_1>0$ away from the origin). 

\begin{figure}
\centering\includegraphics[width=12cm, angle=0]{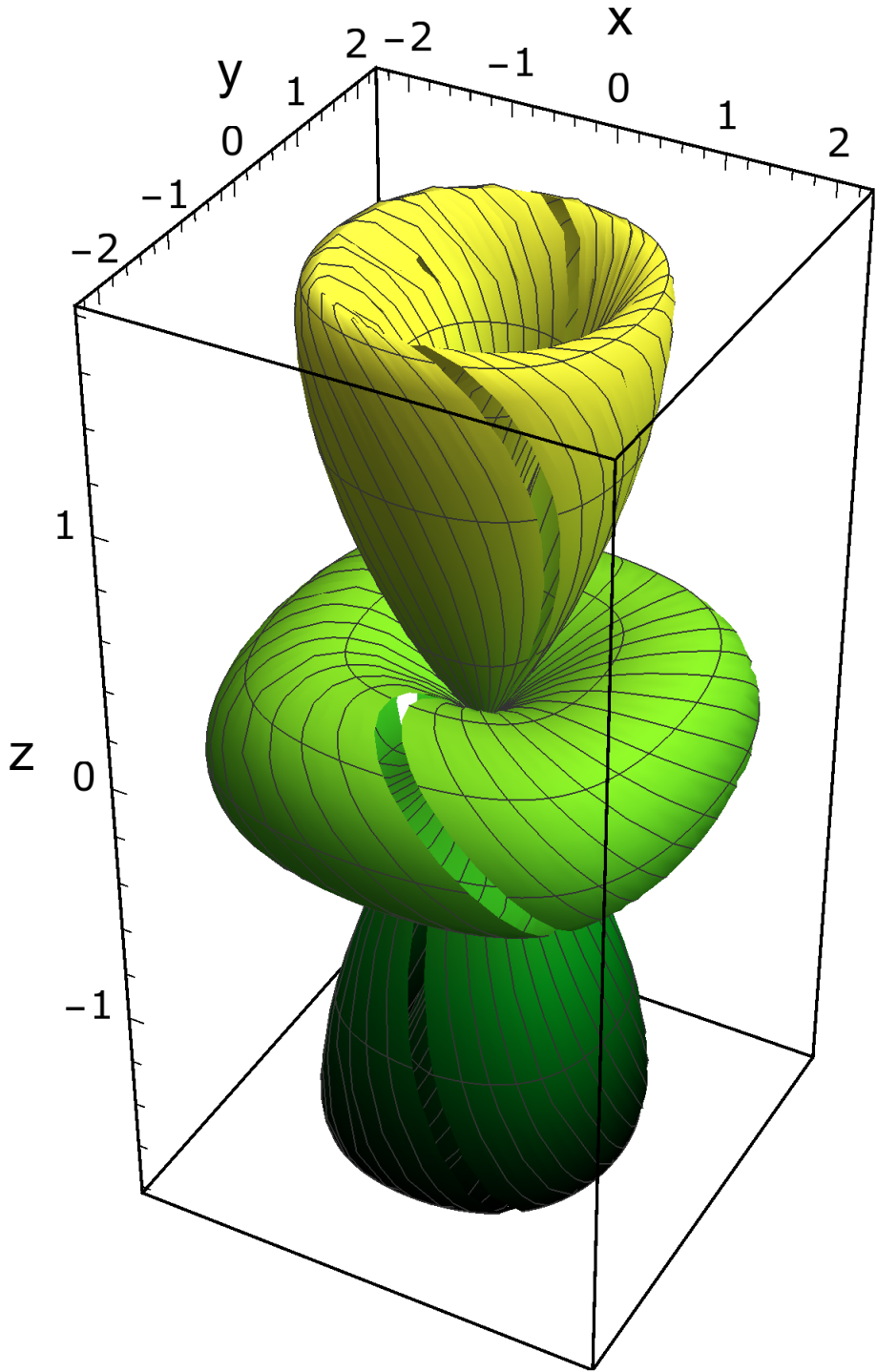}
\caption{The surface $\chi=1$ for the \cite{Low90} solution, where
$\xi$ and $\chi$ are Euler potentials for ${\bf A}$ in (\ref{eq:ll27})
and (\ref{eq:ll28}).
The surface is described by (\ref{eq:ll30a}) et seq. where 
${\bf x}={\bf X}(\mu,\xi)$.$-1<\mu<1$ and $1<\xi<4$. 
The parameters $n=m=1$.}
\label{fig:chi-1-surface}
\end{figure}

A similar strategy can be used to plot the $\chi=const.$ surfaces. 
In the latter case, $\phi=\phi(\theta,\xi)$ and $r=r(\theta)$ 
(note $\chi=const.$) and 
the natural parameters to describe the surface are $(\theta,\xi)$, 
i.e. ${\bf x}={\bf X}(\theta,\xi)$. Note that $r$ is not constant 
on the $\chi=const.$ surface. The magnetic field ${\bf B}=\nabla\chi\times\nabla\xi$ lines lie along the intersections of the $\chi=const.$ and $\xi=const.$ 
surfaces. 

Figure 8 shows the surface $\chi=1$, generated by varying 
$\mu$ and $\xi$ as independent variables in (\ref{eq:ll30}) 
to give the surface in the form ${\bf x}={\bf X}(\mu,\xi)$
where $\mu=\cos\theta$, $-1<\mu<1$ and $1<\xi<4$. The parametes
$n=m=1$. The surface consists of a toroidal doughnut surface
for small $|z|$ and cup like structures which extend along the $z$-axis
both for $z>0$ and for $z<0$.

\begin{figure}
\centering\includegraphics[width=12cm, angle=0]{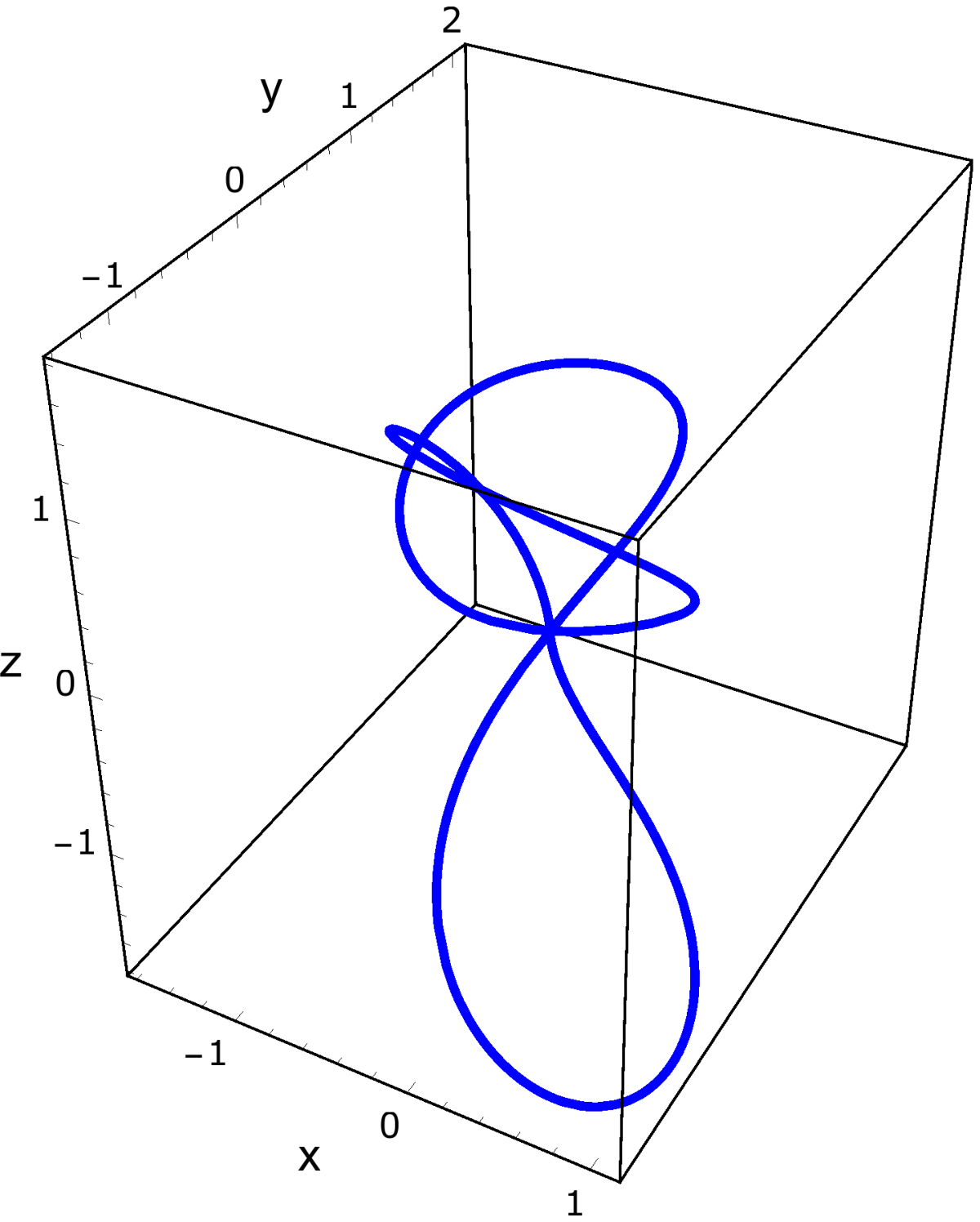}
\caption{The magnetic field line which is described by the intersection 
of the $\xi=1$ and $\chi=1$ Euler potential surfaces, where
$\xi$ and $\chi$ are Euler potentials for ${\bf A}$ for the \cite{Low90} 
force free magnetic field in (\ref{eq:ll27})
and (\ref{eq:ll28}). The parameters $n=m=1$.}
\label{fig:chi1xi1field-line}
\end{figure} 
Figure 9 shows the magnetic field line formed by the intersection 
of the $\xi=1$ and $\chi=1$ Euler potential 
surfaces displayed in Figures 7 and 8. 

\begin{figure}
\centering\includegraphics[width=12cm, angle=0]{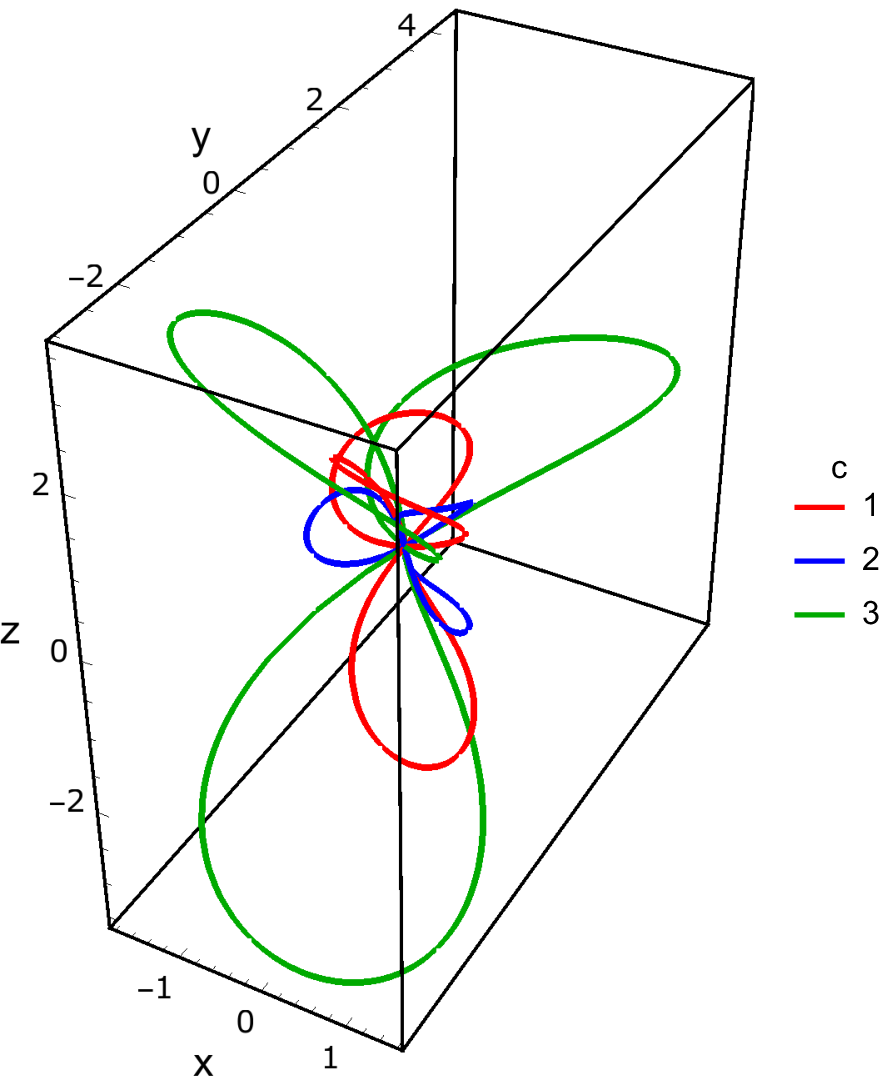}
\caption{The magnetic field lines which are described by the intersection
of the $\xi=c$ and $\chi=c$ Euler potential surfaces, where $c=1,2,3$. 
Here $\xi$ and $\chi$ are Euler potentials for ${\bf A}$ for the \cite{Low90}
force free magnetic field in (\ref{eq:ll27})
and (\ref{eq:ll28}). The parameters $n=m=1$.}
\label{fig:chi-cxi-c field-lines}
\end{figure}
 
Figure 10 shows further examples of magnetic field lines 
formed at the intersection of the $\xi=c$ and $\chi=c$ 
Euler poetential surfaces, for the cases $c=1,2,3$. This is a complicated 
complex of field lines. For ${\bf B}$ to be finite, the origin $r=0$ 
should be excluded, because $B\to\infty$ as $r\to 0$.



\section{Conclusions and Discussion}
In this paper, we studied the Godbillon-Vey invariant which arises in 
magnetohydrodynamics in the case where ${\bf A\bf\cdot B}=0$ where
${\bf B}=\nabla\times{\bf A}$ is the magnetic induction and ${\bf A}$ is the magnetic 
vector potential (\cite{Tur93}, \cite{Webb14a, Webb18}).
 The condition ${\bf A\cdot}(\nabla\times{\bf A})=0$, 
is a necessary and sufficient condition for the Pfaffian ${\bf A\cdot}d{\bf x}=0$
to be integrable (\cite{Sneddon57}), which implies that there exists an 
integrating factor $\mu$ such that $\mu{\bf A\cdot}d{\bf x}=\nabla 
\lambda{\bf\cdot}d{\bf x}=d\lambda$,
for some potential function $\lambda$. This means that the 
 ${\bf B}$
lies on the foliation $\lambda=const.$ and 
the normal to each leaf of the foliation 
is parallel to ${\bf A}$.  The Godbillon-Vey one~-form 
$\boldsymbol{\omega}^1_{\boldsymbol{\eta}}
=\boldsymbol{\eta}{\bf\cdot} d{\bf x}$ arises from the 
requirement that $\boldsymbol{\omega}^1_A={\bf A}{\bf\cdot}d{\bf x}$ satisfies 
$\omega^1_A\wedge d\boldsymbol{\omega}^1_A={\bf A\cdot B}\ d^3x=0$. 
This will be satisfied if there exists a one form $\boldsymbol{\omega}^1_{\boldsymbol{\eta}}$
such that $d\boldsymbol{\omega}_A^1=\boldsymbol{\omega}^1_{\boldsymbol{\eta}}
\wedge \boldsymbol{\omega}^1_A$. Setting 
\begin{equation}
\boldsymbol{\omega}^3_{\boldsymbol{\eta}}=
\boldsymbol{\omega}^1_{\boldsymbol{\eta}} 
 \wedge d\boldsymbol{\omega}^1_{\boldsymbol{\eta}}
=\boldsymbol{\eta}{\bf\cdot}\left(\nabla\times\boldsymbol{\eta}\right)\ d^3x, \label{eq:5.1}
\end{equation}
the above integrability conditions lead to the formulas:
\begin{equation}
\boldsymbol{\eta}=\frac{{\bf A}\times{\bf B}}{|{\bf A}|^2}, 
\quad h_{gv}=\boldsymbol{\eta}{\bf\cdot}(\nabla\times\boldsymbol{\eta}), \quad  
\quad H_{gv}=\int_{V} \boldsymbol{\eta}{\bf\cdot}
(\nabla\times\boldsymbol{\eta})\ d^3x, 
\label{eq:5.2}
\end{equation}
for the Godbillon-Vey field $\boldsymbol{\eta}$ and Godbillon-Vey helicity $H_{gv}$. 

It was shown (proposition 3.1) that if one chooses the electric field 
gauge potential $\psi$ such that $\psi={\bf A\cdot u}$ for which 
the one-form $\boldsymbol{\omega}_A^1={\bf A\cdot}d{\bf x}$ is Lie dragged with the fluid, 
then the Godbillon-Vey helicity density 
$h_{gv}=\boldsymbol{\eta}{\bf\cdot}(\nabla\times\boldsymbol{\eta})$ 
satisfies the conservation law:
\begin{equation}
\deriv{h_{gv}}{t}+\nabla{\bf\cdot}\left({\bf u}h_{gv}+\alpha {\bf B}\right)=0, 
\label{eq:5.3}
\end{equation}
where the scalar parameter $\alpha$ depends on the fluid shear tensor $\sf{\sigma}$ 
via the equations:
\begin{equation}
\alpha=\frac{2 {\bf A}{\bf\cdot}{\sf\sigma}{\bf\cdot}\boldsymbol{\eta}}{|{\bf A}|^2}, 
\quad {\sf\sigma}=\frac{1}{2}\left[\nabla{\bf u}+\left(\nabla{\bf u}\right)^T
-\frac{2}{3}{\sf I}\ \nabla{\bf\cdot u}\right], \label{eq:5.4}
\end{equation}
where ${\sf I}$ is the unit $3\times 3$ dyadic or identity matrix.
 From (\ref{eq:5.3}) it follows that
\begin{equation}
\frac{dH_{gv}}{dt}=0, \label{eq:5.5}
\end{equation}
i.e. $H_{gv}$ is conserved for a volume $V_m$ moving with the flow, 
where it is assumed that  
 $B_n={\bf B\cdot n}$, vanishes on the boundary $\partial V_m$. 
The Godbillon-Vey helicity conservation laws (\ref{eq:5.3}) and (\ref{eq:5.5})
only hold if one uses the advected ${\bf A}$ gauge for ${\bf A}$ 
(e.g. \cite{Gordin87}, \cite{Webb14a}). 
Note that $\alpha=0$ for a shear free flow 
for which ${\sf\sigma}=0$. 

In Section 4, an evolution equation for $h_{gv}$ was 
developed for the case where  
$h_m={\bf A\cdot B}\neq 0$, for which the magnetic field does not lie 
on a foliated
family of surfaces. The Godbillon-Vey helicity density $h_{gv}$ 
defined in (\ref{eq:5.2}) satisfies a modified evolution equation of the form:
\begin{equation}
\deriv{h_{gv}}{t}+\nabla{\bf\cdot}\left({\bf u}h_{gv}\right)=Q,
\label{eq:5.6}
\end{equation}
where $Q$  describes the coupling of the  
magnetic helicity density
($h_m$) with the Godbillon-Vey field $\boldsymbol{\eta}$ via the shear tensor of the flow. 

In ideal, barotropic, incompressible  fluid mechanics,  
 the fluid helicity density:
\begin{equation}
h_f={\bf u\cdot}\boldsymbol{\omega}={\bf u\cdot}(\nabla\times {\bf u}), 
\label{eq:5.9}
\end{equation}
is the analogue of the magnetic helicity 
density $h_m={\bf A\cdot}(\nabla\times{\bf A})$, 
but the analogy is not precise (i.e. there are some caveats 
on the Godbillon-Vey helicity in  the ideal fluid 
context). The condition ${\bf u\cdot}(\nabla\times {\bf u})=0$ implies that 
there is a foliation of the flow, such that $\mu {\bf u\cdot}d{\bf x}=d\lambda$
where $\mu$ is an integrating factor. The fluid vorticity 
$\boldsymbol{\omega}$
lies on the foliation surfaces $\lambda =constant.$,
 and ${\bf u}$ is normal to the surfaces. One can define a 
Godbillon-Vey vector field 
$\boldsymbol{\eta}={\bf u}\times\boldsymbol{\omega}/|{\bf u}|^2$ and set 
$h_\eta=\boldsymbol{\eta}{\bf\cdot}(\nabla\times \boldsymbol{\eta})$ as the 
Godbillon-Vey helicity density. However the equation for ${\bf u}$, 
for incompressible fluid flows, is the momentum or Euler equation:
\begin{equation}
{\bf u}_t-{\bf u}\times \boldsymbol{\omega}
+\nabla\left(p+\frac{1}{2}|{\bf u}|^2\right)=0, \label{eq:5.10}
\end{equation}
where $F=(p+|{\bf u}|^2/2)$ is the Bernoulli function. 
The Euler flow has been the 
subject of many investigations on knotted vortex tubes in fluids. 
Steady solutions of (\ref{eq:5.10}) for  $F=const.$ give rise 
to Beltrami flows, which in most cases give rise to chaotic streamlines
(e.g. The ABC flow is an example: \cite{Dombre86}). 
This is not the exact analogue of the magnetic vector potential equation for 
${\bf A}$ in MHD, namely:
\begin{equation}
{\bf A}_t-{\bf u}\times \left(\nabla\times {\bf A}\right) 
+\nabla({\bf u\cdot A})=0, \label{eq:5.11}
\end{equation}
where we use the advected ${\bf A}{\bf\cdot}d{\bf x}$ gauge.  The net upshot of
this analysis is that one can derive an advection type equation for the 
Godbillon-Vey helicity density $h_{\boldsymbol{\eta}}$ for ideal fluids, 
but in general it is not a conservation law. 


Examples of the nonlinear force free magnetic fields of \cite{Low90} and 
\cite{Prasad14} were illustrated in Section 5. The magnetic field induction
${\bf B}=\nabla\times{\bf A}$ were shown to admit a vector potential ${\bf A}$
which satisfies ${\bf A}{\bf\cdot}{\bf B}=0$, which in turn implies that 
the magnetic field ${\bf B}$ lies on a foliation with normal $\hat{\bf A}
={\bf A}/|{\bf A}|$. The Godbillon-Vey helicity density for the \cite{Low90} 
nonlinear force-free magnetic fields in general is non-zero
(Section 5). 
Note that not all force-free magnetic fields 
have ${\bf A}{\bf\cdot B}=0$. 
Force free magnetic fields 
are widely used in modelling solar magnetic fields in the solar 
chromosphere and corona (e.g. \cite{Sakurai79}, \cite{Wiegelmann12}, 
\cite{Prasad14}). 

\cite{Holm91} studied zero helicity Lagrangian 
kinematics for 3D advection. 
Okhitani (2018) has investigated the 3D Euler equation for incompressible 
fluids, using Clebsch potentials  for zero helicity flows, with the aim in mind 
of elucidating singularity formation in ideal fluids (e.g. he studies both the Taylor Green vortex and the Kida vortex). The role of Godbillon-Vey helicity
in these flows is an interesting possibility for further research. 
\cite{Berger18} have investigated the
absolute magnetic helicity, which uses a poloidal and toroidal decomposition
of the field and uses the Gauss-Bonnet theorem.  
These 
problems pose open questions beyond the scope of the present paper.

\leftline{\bf Acknowledgements}
GMW is supported in part by NASA grant NNX15A165G. SCA is supported in part 
by an NSERC research grant. Q.H. and A.P. acknowledge partial support by 
NASA grant 80NSSC17K0016 and NSF award AGS-1650854. 
GMW acknowledges discussions  
with M.A. Berger on absolute magnetic helicity and toroidal and poloidal 
magnetic field decompositions and on possible applications of Godbillon-Vey
helicity.

\appendix
\section*{Appendix A}
\setcounter{section}{1}

In this appendix we derive the Godbillon-Vey helicity conservation 
equation (\ref{eq:god3.29})
for ideal MHD flows with ${\bf A}{\bf\cdot}{\bf B}=0$. The analysis roughly  
follows that of \cite{Tur93} and \cite{Webb14a}. From \cite{Webb14a}, 
equation (4.95), Faraday's equation (\ref{eq:god2.6}) 
can be written in the form:
\begin{equation}
\left(\derv{t}+{\cal L}_{\bf u}\right) d\boldsymbol{\omega}_A^1\equiv 
\left(\derv{t}+{\cal L}_{\bf u}\right) 
\left(\boldsymbol{\omega}_\eta^1\wedge \boldsymbol{\omega}_A^1\right)=0, 
\label{eq:godA1}
\end{equation}
where $d\boldsymbol{\omega}_A^1={\bf B}{\bf\cdot}d{\bf S}$ is the magnetic flux two form, 
and the decomposition 
$d\boldsymbol{\omega}_A^1=\boldsymbol{\omega}_\eta^1\wedge \boldsymbol{\omega}_A^1$ 
from (\ref{eq:god3.11}) implies ${\bf A}{\bf\cdot}{\bf B}=0$.
 Because we use the advected ${\bf A}$ gauge with $\psi={\bf A}{\bf\cdot}{\bf u}$, 
\begin{equation}
\left(\derv{t}+{\cal L}_{\bf u}\right)\boldsymbol{\omega}_A^1
=\frac{d}{dt}\left({\bf A}{\bf\cdot}d{\bf x}\right)  
=\left[\deriv{\bf A}{t}-{\bf u}\times (\nabla\times{\bf A})
+\nabla({\bf u\cdot A})\right]{\bf\cdot}d{\bf x}=0, \label{eq:godA2}
\end{equation}
is equivalent to the un-curled form of Faraday's equation. 
Taking into account (\ref{eq:godA2}),  (\ref{eq:godA1}) simplifies to:
\begin{equation}
\left[\left(\derv{t}+{\cal L}_{\bf u}\right)\boldsymbol{\omega}_\eta^1\right]
\wedge\boldsymbol{\omega}_A^1=0. \label{eq:godA3}
\end{equation}
Equation (\ref{eq:godA3}) implies:
\begin{equation}
\left(\derv{t}+{\cal L}_{\bf u}\right)\boldsymbol{\omega}_\eta^1=\alpha \boldsymbol{\omega}_A^1,
\label{eq:godA4}
\end{equation}
where the function $\alpha({\bf x},t)$ is yet to be determined. (\ref{eq:godA4}) 
may be written as:
\begin{equation}
\deriv{\boldsymbol{\eta}}{t}- {\bf u}\times(\nabla\times\boldsymbol{\eta})+\nabla ({\bf u\cdot}
\boldsymbol{\eta})=\alpha {\bf A}. \label{eq:godA5}
\end{equation}
Taking the scalar product of (\ref{eq:godA5}) with ${\bf A}$ gives the equation:
\begin{equation}
\alpha |{\bf A}|^2={\bf A}{\bf\cdot}
\left[\deriv{\boldsymbol{\eta}}{t}
-{\bf u}\times(\nabla\times\boldsymbol{\eta})
+\nabla ({\bf u\cdot} \boldsymbol{\eta})\right]. \label{eq:godA6}
\end{equation}
To obtain a simpler formula for $\alpha$, we use the fact 
that $\boldsymbol{\omega}_A^1={\bf A}{\bf\cdot}d{\bf x}$ is Lie dragged 
with the flow in (\ref{eq:godA2}). Taking the scalar product of (\ref{eq:godA2}) with 
$\boldsymbol{\eta}$ gives the equation:
\begin{equation}
0=\boldsymbol{\eta}{\bf\cdot}\left[\deriv{\bf A}{t}-{\bf u}\times (\nabla\times{\bf A})
+\nabla({\bf u\cdot A})\right]. \label{eq:godA7}
\end{equation}
Adding (\ref{eq:godA6}) and (\ref{eq:godA7}) and noting that ${\bf A\cdot}\boldsymbol{\eta}=0$
(note $\boldsymbol{\eta}={\bf A}\times{\bf B}/|{\bf A}|^2$), results in the formula:
\begin{equation}
\alpha=\frac{1}{|{\bf A}|^2} \left[ \boldsymbol{\eta}{\bf\cdot}
\left({\bf u}{\bf\cdot}\nabla{\bf A}+{\bf A}{\bf\cdot}\nabla {\bf u}\right)
+{\bf A}{\bf\cdot}
\left({\bf u}{\bf\cdot}\nabla \boldsymbol{\eta}
+ \boldsymbol{\eta}{\bf\cdot}\nabla{\bf u}\right)\right]. \label{eq:godA8}
\end{equation}
Using the result $\boldsymbol{\eta}{\bf\cdot}{\bf A}=0$, (\ref{eq:godA8}) reduces to:
\begin{equation}
\alpha=\frac{1}{|{\bf A}|^2} \left[ A^s \eta^i\left(\nabla_i u^s+\nabla_s u^i\right)\right]. 
\label{eq:godA9}
\end{equation}
Using the Cauchy-Stokes formula (\cite{Mihalas84}, \cite{Webb94}):
\begin{equation}
u^i_{,j}=\frac{1}{2}\omega_{ij}+\sigma_{ij}+\frac{1}{3}\delta_{ij}\nabla{\bf\cdot}{\bf u}, 
\label{eq:godA10}
\end{equation}
where
\begin{align}
\omega_{ij}=&u^i_{,j}-u^j_{,i}, \nonumber\\
\sigma_{ij}=&\frac{1}{2}\left(u^i_{,j}+u^j_{,i} 
-\frac{2}{3}\delta_{ij}\nabla{\bf\cdot}{\bf u}\right) \label{eq:godA11}
\end{align}
are the rotation tensor ($\omega_{ij}$) and shear tensor ($\sigma_{ij}$) of the flow, 
 (\ref{eq:godA8}) for $\alpha$ reduces to:
\begin{equation}
\alpha=\frac{1}{|{\bf A}|^2}\left({\bf A}{\bf\cdot}{\sf\sigma}{\bf\cdot}\boldsymbol{\eta}
+\boldsymbol{\eta}{\bf\cdot}{\sf\sigma}{\bf\cdot}{\bf A}\right)\equiv 
\frac{2({\bf A}{\bf\cdot}{\sf\sigma}{\bf\cdot}\boldsymbol{\eta})}{|{\bf A}|^2}, 
\label{eq:godA12}
\end{equation}
which is the result (\ref{eq:god3.30}) for $\alpha$. 

Next we show that the Godbillon-Vey helicity 3-form:
\begin{equation}
\boldsymbol{\omega}_\eta^3=\boldsymbol{\omega}_\eta^1\wedge d\boldsymbol{\omega}_\eta^1
\equiv \boldsymbol{\eta}{\bf\cdot}(\nabla\times \boldsymbol{\eta})\ d^3x, \label{eq:godA13}
\end{equation}
satisfies the equation:
\begin{equation}
\left(\derv{t}+{\cal L}_{\bf u}\right)\boldsymbol{\omega}_\eta^3
=-d\left(\alpha d\boldsymbol{\omega}_A^1\right). \label{eq:godA14}
\end{equation}
The result (\ref{eq:godA14}) follows by noting that:
\begin{align}
\left(\derv{t}+{\cal L}_{\bf u}\right)\boldsymbol{\omega}_\eta^3=&
\left[\left(\derv{t}+{\cal L}_{\bf u}\right)\boldsymbol{\omega}_\eta^1\right]
\wedge d\boldsymbol{\omega}_\eta^1 
+\boldsymbol{\omega}_\eta^1\wedge 
\left[\left(\derv{t}+{\cal L}_{\bf u}\right) d\boldsymbol{\omega}_\eta^1\right]\nonumber\\
=&\alpha \boldsymbol{\omega}_A^1\wedge d\boldsymbol{\omega}_\eta^1
+\boldsymbol{\omega}_\eta^1\wedge d\left[\left(\derv{t}+{\cal L}_{\bf u}\right)
\boldsymbol{\omega}_\eta^1\right]\nonumber\\
=&0+\boldsymbol{\omega}_\eta^1\wedge d\left(\alpha \boldsymbol{\omega}_A^1\right)
=-d\left(\boldsymbol{\omega}_\eta^1\wedge\alpha\boldsymbol{\omega}_A^1\right)
=-d\left(\alpha d\boldsymbol{\omega}_A^1\right), \label{eq:godA15}
\end{align}
which proves (\ref{eq:godA14}). In the derivation of (\ref{eq:godA15}), 
the result $\boldsymbol{\omega}_\eta^1\wedge d\boldsymbol{\omega}_A^1=0$ was used, 
which implies $d\boldsymbol{\omega}_\eta^1\wedge \boldsymbol{\omega}_A^1=0$, because:
\begin{equation}
dd\boldsymbol{\omega}_A^1=0=d(\boldsymbol{\omega}_\eta^1\wedge \boldsymbol{\omega}_A^1)
=d \boldsymbol{\omega}_\eta^1\wedge \boldsymbol{\omega}_A^1
-\boldsymbol{\omega}_\eta^1\wedge d\boldsymbol{\omega}_A^1, \label{eq:godA16}
\end{equation}
where we used (\ref{eq:god3.11}). 

From (\ref{eq:godA14}) and (\ref{eq:godA13}), (\ref{eq:godA14}) reduces to:
\begin{equation}
\left(\derv{t}+{\cal L}_{\bf u}\right)
\left[\boldsymbol{\eta}{\bf\cdot}(\nabla\times \boldsymbol{\eta})\ d^3x\right]
=-d\left(\alpha {\bf B}{\bf\cdot}d {\bf S}\right)=-\nabla{\bf\cdot}\left(\alpha {\bf B}\right) 
d^3x. \label{eq:godA17}
\end{equation}
Using Cartan's magic formula, gives:
\begin{equation}
{\cal L}_{\bf u}\left(\boldsymbol{\eta}{\bf\cdot}
(\nabla\times \boldsymbol{\eta})\ d^3x\right) 
={\bf u}\lrcorner d\left(h_{gv}d^3 x\right)+d\left({\bf u}\lrcorner h_{gv}d^3 x\right)
= 0+\nabla{\bf\cdot}\left({\bf u}h_{gv}\right)\ d^3x.
 \label{eq:godA18}
\end{equation}
Using (\ref{eq:godA18}), (\ref{eq:godA17}) reduces to the Godbillon-Vey 
helicity conservation law (\ref{eq:god3.29}).

\appendix
\section*{Appendix B}
\setcounter{section}{2}
\setcounter{equation}{0}
In this appendix we provide a proof of proposition (\ref{prop4.1}) on the form of the 
Godbillon-Vey transport equation described in (\ref{eq:god4.1})-(\ref{eq:god4.4}). 
 Faraday's equation takes the form:
\begin{equation}
\derv{t}\left(\beta {\bf A}+\boldsymbol{\eta}\times {\bf A}\right)
-\nabla\times\left[{\bf u}\times\left(\beta {\bf A}+\boldsymbol{\eta}\times {\bf A}\right)\right]
=0. \label{eq:god4.5}
\end{equation}
Equation (\ref{eq:god4.5}) can then be expressed in the form:
\begin{align}
&\boldsymbol{\eta}\times\left[ {\bf A}_{t} -{\bf u}\times(\nabla\times{\bf A}) 
+\nabla({\bf u}{\bf\cdot}{\bf A})\right]\nonumber\\
&+\left[\boldsymbol{\eta}_t-{\bf u}\times(\nabla\times\boldsymbol{\eta})
+\nabla({\bf u}{\bf\cdot}\boldsymbol{\eta})\right]\times{\bf A}\nonumber\\
&+\left[(\beta{\bf A})_t-\nabla\times({\bf u}\times \beta{\bf A})+{\bf u}\nabla{\bf\cdot}
(\beta {\bf A})\right]=0. \label{eq:god4.6}
\end{align}

The un-curled form of Faraday's equation (\ref{eq:god2.10}) can be 
written in the form:
\begin{equation}
{\bf A}_t-{\bf u}\times(\nabla\times {\bf A})+\nabla({\bf u}{\bf\cdot}{\bf A})=-\nabla\zeta, 
\label{eq:god4.7}
\end{equation}
where
\begin{equation}
\zeta=\psi-{\bf A}{\bf\cdot}{\bf u} \label{eq:god4.8}
\end{equation}
For the advected ${\bf A}$ gauge, $\zeta=0$ and $\psi={\bf A\cdot u}$. 
Substitute of (\ref{eq:god4.7}) into (\ref{eq:god4.6}) gives: 
\begin{align}
&-\boldsymbol{\eta}\times\nabla\zeta +\left[\boldsymbol{\eta}_t-{\bf u}\times(\nabla\times\boldsymbol{\eta})+\nabla({\bf u}{\bf\cdot}\boldsymbol{\eta})\right]\times{\bf A}\nonumber\\
&+\left[(\beta{\bf A})_t-\nabla\times({\bf u}\times \beta{\bf A})+{\bf u}\nabla{\bf\cdot}
(\beta {\bf A})\right]=0. \label{eq:god4.9}
\end{align}
In general, the vectors 
\begin{equation}
{\bf e}_1={\bf A},\quad {\bf e}_2=\boldsymbol{\eta},
\quad {\bf e}_3=\boldsymbol{\eta}\times {\bf A}, \label{eq:god4.10}
\end{equation}
are orthogonal vectors, and in principle, could be used to describe ${\bf B}$ and ${\bf u}$. 

To further simplify (\ref{eq:god4.9}) we make use of the magnetic helicity conservation law (\ref{eq:god3.3}) in the form:
\begin{equation}
\deriv{h_m}{t}+\nabla{\bf\cdot}\left({\bf u} h_m+{\bf B}\zeta\right)=0
\quad \hbox{where}\quad h_m={\bf A\cdot B}, \label{eq:god4.11}
\end{equation}
is the magnetic helicity density. We use the notation:
\begin{align}
{\bf F}=&\boldsymbol{\eta}_t-{\bf u}\times(\nabla\times\boldsymbol{\eta})
+\nabla({\bf u}{\bf\cdot}\boldsymbol{\eta})\label{eq:god4.12}\\
{\bf G}=&(\beta{\bf A})_t-\nabla\times({\bf u}\times \beta{\bf A})+{\bf u}\nabla{\bf\cdot}
(\beta {\bf A}). \label{eq:god4.13}
\end{align}
Using this notation, (\ref{eq:god4.9}) may be written as:
\begin{equation}
-\boldsymbol{\eta}\times\nabla\zeta+{\bf F}\times {\bf A}+{\bf G}=0. \label{eq:god4.14}
\end{equation}

By using (\ref{eq:god4.11}), the expression for ${\bf G}$ 
reduces to:
\begin{equation}
{\bf G}=-\frac{\bf A}{|{\bf A}|^2} {\bf B}{\bf\cdot}\nabla\zeta +\frac{h_m}{|{\bf A}|^2} 
\left(\frac{d{\bf A}}{dt}-{\bf A}{\bf\cdot}\nabla {\bf u}
- 2\hat{\bf A}\hat{\bf A}{\bf\cdot}
\frac{d{\bf A}}{dt}\right), \label{eq:god4.15}
\end{equation}
where $\hat{\bf A}={\bf A}/|{\bf A}|$ and 
$d{\bf A}/dt={\bf A}_t+{\bf u}{\bf\cdot}\nabla {\bf A}$.
To further reduce (\ref{eq:god4.15}) we use the identity:
\begin{equation}
\nabla({\bf A\cdot u})={\bf A}{\bf\cdot}\nabla {\bf u}+{\bf u}{\bf\cdot}\nabla{\bf A}
+{\bf u}\times (\nabla\times{\bf A})+{\bf A}\times (\nabla\times{\bf u}), \label{eq:god4.16}
\end{equation}
in the un-curled form of Faraday's equation (\ref{eq:god4.7}) to obtain:
\begin{equation}
\frac{d{\bf A}}{dt}+{\bf A}{\bf\cdot}\nabla {\bf u}+{\bf A}\times (\nabla\times{\bf u})
+\nabla\zeta=0. \label{eq:god4.17}
\end{equation}
 
${\bf G}$ can be split up into components perpendicular and parallel to ${\bf A}$ as
${\bf G}={\bf G}_\parallel +{\bf G}_\perp$, by noting that:
\begin{equation}
{\bf G}= -\frac{\bf A}{|{\bf A}|^2} {\bf B}{\bf\cdot}\nabla\zeta 
+\frac{h_m}{|{\bf A}|^2} 
\left[\left({\sf I}-\hat{\bf A}\hat{\bf A}\right)
\left(\frac{d{\bf A}}{dt} -{\bf A}{\bf\cdot}\nabla {\bf u}\right)
-\hat{\bf A}\hat{\bf A}{\bf\cdot}\left(\frac{d{\bf A}} {dt}
+{\bf A}{\bf\cdot}\nabla {\bf u}\right)\right], \label{eq:god4.18}
\end{equation}
where the projection tensor ${\sf P}_A=({\sf I}-\hat{\bf A}\hat{\bf A})$ annuls 
vectors parallel to ${\bf A}$, i.e. ${\sf P}_A{\bf A}=0$. 
Thus, ${\bf G}_\parallel$ and ${\bf G}_\perp$ are given by:
\begin{align}
{\bf G}_\parallel=&-\frac{\hat{\bf A}}{|{\bf A}|} {\bf B\cdot}\nabla\zeta
- \frac{h_m}{|{\bf A}|^2} \hat{\bf A}\hat{\bf A}{\bf\cdot}
\left(\frac{d{\bf A}}{dt}+{\bf A}{\bf\cdot}\nabla {\bf u}\right), \label{eq:god4.19}\\
{\bf G}_\perp=&\frac{h_m}{|{\bf A}|^2}\left({\sf I}-\hat{\bf A}\hat{\bf A}\right)
{\bf\cdot}\left(\frac{d{\bf A}}{dt} -{\bf A}{\bf\cdot}\nabla {\bf u}\right). 
\label{eq:god4.20}
\end{align}

Using $d{\bf A}/dt$ from (\ref{eq:god4.17}) in (\ref{eq:god4.19}) gives:
\begin{equation}
{\bf G}_\parallel= -\frac{\bf A}{|{\bf A}|^2} {\bf B}_\perp {\bf\cdot}\nabla\zeta
=-\frac{\bf A}{|{\bf A}|^2}(\boldsymbol{\eta}\times{\bf A}) {\bf\cdot}\nabla\zeta. 
\label{eq:god4.21}
\end{equation}
Similarly, (\ref{eq:god4.20}) reduces to:
\begin{equation}
{\bf G}_\perp= \frac{h_m}{|{\bf A}|^4} {\bf A}
\times\left[{\bf A}\times\left(2 {\bf A}{\bf\cdot}\nabla{\bf u}
+{\bf A}\times\boldsymbol{\omega}+\nabla\zeta\right)\right], \label{eq:god4.22}
\end{equation}
where $\boldsymbol{\omega}=\nabla\times{\bf u}$ is the fluid vorticity. Using the 
Cauchy-Stokes formula (\ref{eq:godA10}) results in the formula:
\begin{equation}
{\bf A}{\bf\cdot}\nabla {\bf u}=\frac{1}{2}\boldsymbol{\omega}\times {\bf A}
+\sf{\sigma}{\bf\cdot}{\bf A} +\frac{1}{3} (\nabla{\bf\cdot}{\bf u}) {\bf A}. 
\label{eq:god4.23}
\end{equation}
Substituting (\ref{eq:god4.23}) in (\ref{eq:god4.22}) gives the formula:
\begin{equation}
{\bf G}_\perp=\frac{h_m}{|{\bf A}|^4} {\bf A}
\times\left[{\bf A}\times\left(2 {\sf\sigma}{\bf\cdot}{\bf A}+\nabla\zeta\right)\right]. 
\label{eq:god4.24}
\end{equation}

Taking the scalar product of (\ref{eq:god4.14}) with $\hat{\bf A}$ results in the equation:
\begin{equation}
\left({\bf G}_\parallel-\boldsymbol{\eta}\times \nabla\zeta\right){\bf\cdot}\hat{\bf A}=0. 
\label{eq:god4.25}
\end{equation}
Using (\ref{eq:god4.21}) for ${\bf G}_\parallel$ in (\ref{eq:god4.25}), results in the 
balance equation:
\begin{equation}
-\frac{\nabla\zeta}{|{\bf A}|} {\bf\cdot}\left(\boldsymbol{\eta}\times {\bf A}
+{\bf A}\times \boldsymbol{\eta}\right)=0, \label{eq:god4.26}
\end{equation}
which is identically satisfied. 

The component of (\ref{eq:god4.14}) perpendicular to ${\bf A}$ gives the vector equation:
\begin{equation}
{\bf G}_\perp+{\bf F}\times {\bf A}+ \frac{\bf A}{|{\bf A}|^2} 
\times \left[{\bf A}\times 
(\boldsymbol{\eta}\times\nabla\zeta)\right]=0, 
\label{eq:god4.27}
\end{equation}
which reduces to the equation:
\begin{equation}
\left\{{\bf F}-\frac{h_m}{|{\bf A}|^4} {\bf A}\times \left[ 2({\sf\sigma}{\bf\cdot}{\bf A})
+\nabla\zeta\right] -\frac{{\bf A}\times(\boldsymbol{\eta}\times\nabla\zeta)}
{|{\bf A}|^2}\right\}\times{\bf A}=0. 
\label{eq:god4.28}
\end{equation}
Equation (\ref{eq:god4.28}) is satisfied if:
\begin{equation}
{\bf F}-\frac{h_m}{|{\bf A}|^4} {\bf A}\times \left[ 2({\sf\sigma}{\bf\cdot}{\bf A})
+\nabla\zeta\right] -\frac{{\bf A}\times(\boldsymbol{\eta}\times\nabla\zeta)}
{|{\bf A}|^2}=\alpha{\bf A}, \label{eq:god4.29}
\end{equation}
where $\alpha({\bf x},t)$ is a scalar function of ${\bf x}$ and $t$, which is yet to 
be determined. 

Using (\ref{eq:god4.29}) 
and (\ref{eq:god4.12}) for ${\bf F}$ 
gives the equation:
\begin{equation}
{\bf A}{\bf\cdot}\left[\boldsymbol{\eta}_t-{\bf u}\times(\nabla\times\boldsymbol{\eta})
+\nabla({\bf u}{\bf\cdot}\boldsymbol{\eta})\right]=\alpha |{\bf A}|^2. 
\label{eq:god4.30}
\end{equation}
To obtain a more useful form for $\alpha$ take the 
dot product of (\ref{eq:god4.7})  with $\boldsymbol{\eta}$ gives the 
equation:
\begin{equation}
\boldsymbol{\eta}{\bf\cdot}\left[{\bf A}_t-{\bf u}\times(\nabla\times{\bf A})
+\nabla({\bf u\cdot A})+\nabla\zeta\right]=0. \label{eq:god4.31}
\end{equation}
Adding (\ref{eq:god4.30}) and (\ref{eq:god4.31}), and using  
 ${\bf A\cdot}\boldsymbol{\eta}=0$, we obtain the 
equation:
\begin{equation}
\alpha=\frac{1}{|{\bf A}|^2}\left[ \boldsymbol{\eta}{\bf\cdot}
\left({\bf u}{\bf\cdot}\nabla{\bf A}+{\bf A}{\bf\cdot}\nabla{\bf u}\right)
+{\bf A}{\bf\cdot}\left({\bf u}{\bf\cdot}\nabla\boldsymbol{\eta}
+\boldsymbol{\eta}{\bf\cdot}\nabla{\bf u}\right)+\boldsymbol{\eta}{\bf\cdot}\nabla\zeta\right].
\label{eq:god4.32}
\end{equation}
Using $\boldsymbol{\eta}{\bf\cdot}{\bf A}=0$, (\ref{eq:god4.32}) can be written in the form:
\begin{equation}
\alpha=\frac{1}{|{\bf A}|^2}\left[{\bf A}\boldsymbol{\eta}
:\left(\nabla {\bf u}+(\nabla {\bf u})^T\right)+\boldsymbol{\eta}{\bf\cdot}\nabla\zeta\right].
\label{eq:god4.33}
\end{equation}
Using the Cauchy-Stokes formula (\ref{eq:godA10}) then gives expression (\ref{eq:god4.2}) 
for $\alpha$, namely:
\begin{equation}
\alpha=\frac{\left(2 {\bf A}{\bf\cdot}{\sf\sigma}{\bf\cdot}\boldsymbol{\eta}
+\boldsymbol{\eta}{\bf\cdot}\nabla\zeta\right)}
{|{\bf A}|^2}. \label{eq:god4.34}
\end{equation}

Next, using (\ref{eq:god4.12}) for ${\bf F}$, (\ref{eq:god4.29}) reduces to the equation:
\begin{equation}
\boldsymbol{\eta}_t-{\bf u}\times(\nabla\times\boldsymbol{\eta})
+\nabla({\bf u}{\bf\cdot}\boldsymbol{\eta})={\bf S}, \label{eq:god4.35}
\end{equation}
where the source term ${\bf S}$ is given by (\ref{eq:god4.2}). 

To obtain (\ref{eq:god4.1}), take the curl of 
(\ref{eq:god4.35}) to obtain the equation:
\begin{equation}
(\nabla\times\boldsymbol{\eta})_t-\nabla\times\left[{\bf u}
\times(\nabla\times\boldsymbol{\eta})\right]=\nabla\times{\bf S}. \label{eq:god4.36}
\end{equation}
Take the scalar product of (\ref{eq:god4.35}) with $\nabla\times\boldsymbol{\eta}$ and add 
the resultant equation to the scalar product of (\ref{eq:god4.36}) with $\boldsymbol{\eta}$ 
to obtain the Godbillon-Vey helicity transport equation (\ref{eq:god4.1}).

\appendix
\section*{Appendix C}
\setcounter{section}{3}
\setcounter{equation}{0}
 
In this appendix we discuss Clebsch potentials representations for 
${\bf A}$ in calculating the Godbillon-Vey helicity
$h_{\boldsymbol{\eta}}=\boldsymbol{\eta}{\bf\cdot}(\nabla\times
\boldsymbol{\eta})$ of Section 3.
 If we choose the Clebsch representation (\ref{eq:god3.25}):
\begin{equation}
{\bf A}=\nu\nabla\lambda+\nabla\phi,\quad {\bf B}=\nabla\times {\bf A}
=\nabla\nu\times\nabla\lambda, \label{eq:C1}
\end{equation}
we obtain:
\begin{equation}
{\bf A\cdot B}=\left(\nu\nabla\lambda+\nabla\phi\right){\bf\cdot}
(\nabla\nu\times\nabla\lambda) = \nabla\phi{\bf\cdot}
(\nabla\nu\times\nabla\lambda)
=\frac{\partial(\phi,\nu,\lambda)}{\partial(x,y,z)}. \label{eq:C2}
\end{equation}
Thus, ${\bf A\cdot B}=0$ if $\phi=\phi(\nu,\lambda)$. 

The Godbillon-Vey field $\boldsymbol{\eta}$ defined in (\ref{eq:god3.20}) 
is given by:
\begin{equation}
\boldsymbol{\eta}=\frac{{\bf A}\times {\bf B}}{|{\bf A}|^2}
=(\nu\nabla\lambda+\nabla\phi)\times(\nabla\nu\times\nabla\lambda)/A^2
=\eta_1{\bf e}^1+\eta_2{\bf e}^2, \label{eq:C3}
\end{equation}
where we use the notation:
\begin{equation}
{\bf e}^1=\nabla\nu,\quad {\bf e}^2=\nabla\lambda.
 \label{eq:C4}
\end{equation}
Below, we obtain a third independent Clebsch variable $\gamma$. 
 The Clebsch variables
$\nu$, $\lambda$ and $\gamma$ are independent Lagrange labels. 

From \cite{Golovin11} the Lie derivative operators:
\begin{equation}
X_1=\frac{d}{dt}=\derv{t}+{\bf u\cdot}\nabla, \quad X_2={\bf b}\equiv
\frac{\bf B}{\rho}{\bf\cdot}\nabla, \label{eq:C4a}
\end{equation}
commute  because of the frozen in field theorem
and  the mass continuity equation. Thus:
\begin{equation}
\left[\frac{d}{dt},{\bf b}\right]
=\left[\derv{t}+{\bf u}{\bf\cdot}\nabla, 
{\bf b}{\bf\cdot}\nabla\right]\equiv [X_1,X_2]=0.  \label{eq:C4b}
\end{equation}
Condition (\ref{eq:C4b}) implies that $X_1$ and $X_2$  form a 2D 
Lie algebra. The integrabilty conditions (\ref{eq:C4b})  by Frobenius theorem,
implies that $X_1$ and $X_2$ have the representations:
\begin{equation}
X_1=\frac{d}{dt}\equiv \left(\derv{t}\right)_{{\bf x}_0}, \quad 
X_2=\left(\derv{\gamma}\right)_t, \label{eq:C4c}
\end{equation}
where ${\bf x}_0$ correspond to  $\nu$, 
$\lambda$, and $\gamma$ which are  advected with the flow, i.e. 
\begin{equation}
\frac{d\nu}{dt}=\frac{d\lambda}{dt}=\frac{d\gamma}{dt}=0. \label{eq:C4d}
\end{equation}
Using the Lagrangian map:
\begin{equation}
x^s=x^s\left(t,\nu,\lambda,\gamma\right)=(t,x,y,z), \quad s=0,1,2,3, 
\label{eq:C4e}
\end{equation}
and using the notation:
\begin{equation}
(\xi^1,\xi^2,\xi^3)=(\nu,\lambda,\gamma), \label{eq:C4ea}
\end{equation}
for the independent Lagrange labels $\nu$, $\lambda$, $\gamma$,
it follows that:
\begin{equation}
{\bf e}^i\times {\bf e}^j=\frac{\epsilon_{ijk}}{\sqrt{g}} {\bf e}_k, 
\quad {\bf e}_i\times{\bf e}_j=\sqrt{g}\epsilon_{ijk} {\bf e}^k, \label{eq:C4f}
\end{equation}
where ${\bf e}_i=\partial {\bf x}/\partial \xi^i$ is the basis that is dual
to the base $\{{\bf e}^i\}$, 
i. e. $\langle {\bf e}^i, {\bf e}_j\rangle=\delta^i_j$. The metric tensor 
${\sf g}$ has covariant ($g_{ij}$) and contravariant ($g^{ij}$)  components
defined by 
\begin{align}
g_{ij}=&{\bf e}_i{\bf\cdot}{\bf e}_j, 
\quad g^{ij}={\bf e}^i{\bf\cdot}{\bf e}^j,
\quad g=\det\left(g_{ij}\right)=J^2, 
\nonumber\\
J=&\det\left(\frac{\partial x^i}{\partial\xi^j}\right)
={\bf e}_1 {\bf\cdot}\left({\bf e}_2\times {\bf e}_3\right)=\sqrt{g}, 
\label{eq:C4g}
\end{align}
(e.g. \cite{Boozer04}).

Note from (\ref{eq:C4f}) that:
\begin{equation}
{\bf B}={\bf e}^1\times{\bf e}^2=\frac{{\bf e}_3}{\sqrt{g}}. \label{eq:C5}
\end{equation}
The coefficients $\eta_1$ and $\eta_2$ in (\ref{eq:C3}) are given by:
\begin{align}
\eta_1=&\left[(\nu+\phi_\lambda) g^{22}
+\phi_\nu g^{12}+\phi_{\gamma} g^{32}\right]/|{\bf A}|^2, 
\nonumber\\
\eta_2=&-\left[(\nu+\phi_\lambda)g^{21}+\phi_\nu g^{11}+\phi_\gamma g^{31}
\right]/|{\bf A}|^2. \label{eq:C6}
\end{align}
Note that the Lagrangian mass continuity equation
$\rho d^3x=\rho_0 d^3x_0$ reduces to $\rho J\equiv \rho\sqrt{g}=\rho_0$.
Choosing $\rho_0=1$ we find:
\begin{equation}
X_2=\frac{{\bf B}}{\rho}{\bf\cdot}\nabla=\frac{{\bf e}_3}
{\sqrt{g}\rho}{\bf\cdot}\nabla={\bf e}_3{\bf\cdot}\nabla=\derv{\gamma}, 
\label{eq:C7}
\end{equation}
which verifies (\ref{eq:C4c}). 

A straightforward calculation gives:
\begin{equation}
h_{\boldsymbol{\eta}}=\boldsymbol{\eta}{\bf\cdot}
(\nabla\times\boldsymbol{\eta}) =\left({\bf e}^1\times{\bf e}^2\right)
{\bf\cdot}\left[\eta_2\nabla\eta_1-\eta_1\nabla\eta_2\right]
={\bf B}{\bf\cdot} \left(\eta_2\right)^2\nabla\left(\eta_1/\eta_2\right). 
\label{eq:C8}
\end{equation}
In the special case where $\phi=0$ (\ref{eq:C8}) simplifies to:
\begin{equation}
h_{\boldsymbol{\eta}}=-\left(\eta_2\right)^2 
{\bf B\cdot}\nabla\left(g^{22}/g^{21}\right). \label{eq:C9}
\end{equation}

If $\phi_\gamma=0$, i.e. $\phi=\phi(\nu,\lambda)$ then ${\bf A\cdot B}=0$ 
and the space is then foliated 
(\cite{Reinhart73} and \cite{Rovenski18, Rovenski19}). 

\appendix
\section*{Appendix D}
\setcounter{section}{4}
\setcounter{equation}{0}
From (\ref{eq:god3.25}), the condition
\begin{equation}
\tilde{\bf A}={\bf A}+\nabla\phi=\nu\nabla\lambda+\nabla\phi=\tilde{\nu}\nabla\tilde{\lambda}, 
\label{eq:D1}
\end{equation}
for a gauge transformation will be satisfied (we assume $\tilde{\lambda}$ 
and $\tilde{\nu}$ are functions of $\lambda$ and $\nu$) if: 
\begin{equation}
\tilde{\nu}\deriv{\tilde\lambda}{\lambda}=\nu+\phi_{\lambda}, \quad 
\tilde{\nu}\deriv{\tilde\lambda}{\nu}=\phi_\nu. \label{eq:D2}
\end{equation}
The integrability conditions of (\ref{eq:D2}) are: 
\begin{equation}
\frac{\partial^2{\tilde\lambda}}{\partial\lambda\partial\nu}=
\frac{\partial^2{\tilde\lambda}}{\partial\nu\partial\lambda}. \label{eq:D3}
\end{equation}
 The integrability equations (\ref{eq:D3}) are satisfied
if $\tilde{\nu}$ satisfies the first order partial differential equation:
\begin{equation}
\deriv{\tilde\nu}{\lambda}\deriv{\phi}{\nu} +\deriv{\tilde\nu}{\nu}
\left(-\deriv{\phi}{\lambda}-\nu\right)+\tilde{\nu}=0. 
\label{eq:D4}
\end{equation}
The first order partial differential equation for $\tilde{\nu}$ may be solved 
in principle by integrating the characteristics:
\begin{equation}
\frac{d\lambda}{\phi_{\nu}}=\frac{d\nu}{-\phi_\lambda-\nu}
=-\frac{d\tilde{\nu}}{\tilde{\nu}}, \label{eq:D5}
\end{equation}
(\cite{Sneddon57}). After the solution of (\ref{eq:D4})-(\ref{eq:D5}) 
 is established, 
the solution for $\tilde{\lambda}$ can be obtained by integrating, 
the guaranteed integrable equation system  (\ref{eq:D2}). 

\appendix
\section*{Appendix E}
\setcounter{section}{5}
\setcounter{equation}{0}

In this appendix we obtain the \cite{Reinhart73} 
version of the Godbillon Vey helicity of a co-dimension 1 foliation 
in 3D geometry (see also \cite{Rovenski18,Rovenski19}). 
We use both differential forms and more classical approaches to 
the geometry of foliations in our analysis. 
The \cite{Reinhart73} formula  
could in principle  be obtained by using the method of moving frames 
(e.g. \cite{Flanders63} Chapter 4).

The Godbillon-Vey invariant $h_{gv}$ is defined as:
\begin{equation}
H_{gv}=\int_{V_m} \eta\wedge d\eta, \label{eq:E1}
\end{equation}
where $\eta$ is the Godbillon-Vey 1-form  defined below.

The Serret-Frenet equations for the normal curve to the 
surface $\Phi({\bf x})=const.$ of the foliation have the form:
\begin{equation}
\nabla_{\bf T}{\bf T}=\kappa {\bf N}, \quad 
\nabla_{\bf T}{\bf N}= -\kappa {\bf T}
+\tau {\bf B}, \quad \nabla_{\bf T}{\bf B}=-\tau {\bf N}, \label{eq:E2}
\end{equation}
where
\begin{equation}
\nabla_{\bf T}=\frac{d}{ds}={\bf T}{\bf\cdot}\nabla, \label{eq:E3}
\end{equation}
 is the directional derivative along the tangent vector to the 
normal curve (i.e. ${\bf T}$ is the normal to each of the surfaces 
of the foliation $\Phi({\bf x})=const.)$.
 Here we assume ${\bf A}{\bf\cdot}\nabla\times{\bf A}=0$ from which 
it follows that ${\bf A}{\bf\cdot}d{\bf x}=0$
is integrable, i.e. there exists an integrating factor $\mu$ where 
$\mu {\bf A}=\nabla\Phi$ and 
\begin{equation}
{\bf T}=\hat{\bf A}=\frac{\bf A}{A}
\equiv \frac{\nabla{\Phi}}{|\nabla{\Phi}|} \quad\hbox{and}\quad A=|{\bf A}|.
\label{eq:E4}
\end{equation}
The base vectors $({\bf T},{\bf N}, {\bf B})$ form an orthonormal triad 
where
\begin{equation}
{\bf T}\times{\bf N}={\bf B}, \label{eq:E5}
\end{equation}
where ${\bf N}$ is the principal normal and ${\bf B}$ is the binormal 
to the curve.
$\kappa$ and $\tau$ are the curvature and torsion of the curve. 
To simplify the notation 
 we write (\ref{eq:E2}) in the form:
\begin{equation}
\nabla_{{\bf e}_3} {\bf e}_3=\kappa {\bf e}_1,
\quad \nabla_{{\bf e}_3} {\bf e}_1=-\kappa {\bf e}_3+\tau{\bf e}_2, 
\quad \nabla_{{\bf e}_3} {\bf e}_2=-\tau {\bf e}_1, \label{eq:E6}
\end{equation}
where $({\bf e}_1,{\bf e}_2,{\bf e}_3)\equiv ({\bf N},{\bf B},{\bf T})$. 
We could use a more general orthonormal triad 
$({\bf d}_1,{\bf d}_2,{\bf d}_3)$ to frame the curve, which does not have 
ambiguity if the curve is a straight line (e.g. \cite{Bishop75}). 
We use the standard notation 
\begin{equation}
g_{ij}={\bf e}_i{\bf\cdot}{\bf e}_j=g({\bf e}_i,{\bf e}_j), \quad
g^{ij}={\bf e}^i{\bf\cdot}{\bf e}^j, \label{eq:E7}
\end{equation}
for the covariant ($g_{ij}$) and contravariant ($g^{ij}$) components of the 
metric tensor. $\{{\bf e}^i\}$ is dual to the base 
$\left\{{\bf e}_i\right\}$, 
i.e. $\left\langle{\bf e}^i,{\bf e}_j\right\rangle =\delta^i_j$ where 
$\delta^i_j$ is the Kronecker-delta symbol. We also use the affine connection 
formulae:
\begin{equation}
\deriv{{\bf e}_i}{q^j}\equiv ({\bf e}_j{\bf\cdot}\nabla){\bf e}_i
=\Gamma^s_{ij}{\bf e}_s, \quad \deriv{{\bf e}^i}{q^j}=-\Gamma^i_{sj} {\bf e}^s,   \label{eq:E8}
\end{equation}
where $\Gamma^s_{ij}$ are the affine connection coefficients. Because  
 $({\bf e}_1,{\bf e}_2,{\bf e}_3)$ are orthonormal we obtain:
\begin{equation}
g_{ij}=\delta_{ij}={\bf e}_i{\bf\cdot}{\bf e}_j,
\quad g^{ij}={\bf e}^i{\bf\cdot}{\bf e}^j=\delta^{ij}, \label{eq:E9}
\end{equation}
The coordinates $\{q^i\}$ are local and not global 
coordinates, but they suffice for the local description of the 
foliation.
Differentiation of (\ref{eq:E9}) with respect to the $q^a$ gives the equations:
\begin{equation}
\left(\nabla_{{\bf e}_a}{\bf e}_i\right){\bf\cdot}{\bf e}_j
+{\bf e}_i{\bf\cdot}\left(\nabla_{{\bf e}_a}{\bf e}_j\right)=0, \label{eq:E10}
\end{equation}
which using (\ref{eq:E8}) reduces to the relations:
\begin{equation}
\Gamma^j_{ia}+\Gamma^i_{ja}=0. \label{eq:E11}
\end{equation}

Using (\ref{eq:E8})-(\ref{eq:E11}) we obtain the results:
\begin{align}
\Gamma^1_{33}=&\kappa,\quad \Gamma^2_{33}=\Gamma^3_{33}=0, \nonumber\\
\Gamma^3_{13}=&-\kappa,\quad \Gamma^2_{13}=\tau,\quad \Gamma^1_{13}=0, 
\nonumber\\
\Gamma^1_{23}=&-\tau,\quad \Gamma^2_{23}=0. \label{eq:E12}
\end{align}
The second fundamental form ${\rm II}$ for the surface is given by:
\begin{align}
{\rm II}=&d^2{\bf x}{\bf\cdot}{\bf e}_3=\Gamma^3_{\alpha\beta}dq^\beta dq^\alpha\nonumber\\
=&\Gamma^3_{11}(dq^1)^2+\Gamma^3_{22} (dq^2)^2 +\left(\Gamma^3_{12}
+\Gamma^3_{21}\right) dq^1 dq^2\nonumber\\
\equiv& h_{11}(dq^1)^2 +h_{22}(dq^2)^2+(h_{21}+h_{12})dq^1 dq^2, \label{eq:E13}
\end{align}
where 
\begin{equation}
h_{ij}=\Gamma^3_{ji}\equiv g(\nabla_{{\bf e}_i} {\bf e}_j, {\bf e}_3), \quad (i,j=1,2), \label{eq:E14}
\end{equation}
define the coefficients for the second fundamental form 
(e.g. \cite{Lipschutz69}). 

The Godbillon-Vey one form is given by:
\begin{equation}
\eta=\kappa {\bf e}^1, \label{eq:E15}
\end{equation}
(e.g. \cite{Rovenski18,Rovenski19}). 
Taking the exterior derivative of (\ref{eq:E15}) gives:
\begin{equation}
d\eta=-\kappa\Gamma^1_{ps}
 {\bf e}^p\wedge {\bf e}^s+\deriv{\kappa}{q^s} {\bf e}^1\wedge {\bf e}^s. 
\label{eq:E16}
\end{equation}
The Godbillon-Vey 3-form is given by
\begin{align}
h_{gv}^{RW} dV_g=&\eta\wedge d\eta=
-\kappa^2 \Gamma^1_{ps}{\bf e}^1\wedge {\bf e}^p\wedge {\bf e}^s
=-\kappa^2\left(\Gamma^1_{23}-\Gamma^1_{32}\right)dV_g\nonumber\\
=&\kappa^2\left(\tau-h_{21}\right) dV_g
\equiv \kappa^2\left(\tau-h_{BN}\right) dV_g, \label{eq:E17}
\end{align}
where the superscript ${RW}$ in (\ref{eq:E17}) refers to \cite{Reinhart73}, 
and  
\begin{equation}
dV_g=dq^1\wedge dq^2\wedge dq^3
\equiv {\bf e}^1\wedge{\bf e}^2\wedge {\bf e}^3, \label{eq:E18}
\end{equation}
is the volume element for the 3-form (\ref{eq:E17}). The Godbillon-Vey 3-form
(\ref{eq:E17}) is the formula given by \cite{Reinhart73} and 
\cite{Rovenski18,Rovenski19} 
 (note $h_{21}=h_{BN}$ in \cite{Rovenski18, 
Rovenski19}). 

The Godbillon-Vey 3-form (\ref{eq:E17}) is  equivalent 
to the Godbillon 
Vey 3-form used in the present paper in the sense that the helicity density 
$h_{gv}^{RW}\equiv h_{gv}$ where $h_{gv}$ is the Godbillon-Vey helicity 
density used in the present paper, modulo a pure divergence 
term, i.e., $h_{gv}^{RW}=h_{gv}+\nabla{\bf\cdot R}$ (see below). 
The differences of these 2 forms are described below.   
Following \cite{Reinhart73} and \cite{Rovenski18, Rovenski19} 
we first identify a one form $\omega$ that is dual to the normal
 ${\bf T}$ to the foliation, such that
\begin{equation}
\boldsymbol{\omega}({\bf T})\equiv{\bf T}\lrcorner \boldsymbol{\omega}=1. 
 \label{eq:E19}
\end{equation}
The analog of the Serret-Frenet equation for ${\bf T}$ in (\ref{eq:E2})
using the dual one-form $\boldsymbol{\omega}$ is given by Cartan's 
magic formula:
\begin{equation}
{\cal L}_{\bf T}(\boldsymbol{\omega})
={\bf T}\lrcorner d\boldsymbol{\omega} 
+d ({\bf T}\lrcorner \boldsymbol{\omega})\equiv {\bf T}\lrcorner 
d\boldsymbol{\omega}, \label{eq:E20}
\end{equation}
because ${\bf T}\lrcorner \boldsymbol{\omega}=1$. 

There is some freedom in the 
 choice of $\omega_i$ and $T^i$ in (\ref{eq:E19}). 
For example if we choose:
\begin{equation}
{\bf T}=\frac{{\bf A}}{A^2},\quad \hbox{then}\quad 
\boldsymbol{\omega}={\bf A}{\bf\cdot}d{\bf x}. \label{eq:E21}
\end{equation}
Here $A^2=|{\bf A}|^2$. We find:
\begin{align}
d\boldsymbol{\omega}=&{\bf B}{\bf\cdot}d{\bf S}=B_x dy\wedge dz
+B_y dz\wedge dx+B_z dx\wedge dy, \nonumber\\
{\cal L}_{\bf T}(\boldsymbol{\omega})=& \frac{\bf A}{A^2} 
\lrcorner ({\bf B}{\bf\cdot}d{\bf S})=-\frac{({\bf A}\times{\bf B})
{\bf\cdot}d{\bf x}}{A^2}
=-\boldsymbol{\eta}{\bf\cdot}d{\bf x}, \label{eq:E22}
\end{align}
where
\begin{equation}
\boldsymbol{\eta}=\frac{({\bf A}\times{\bf B})}{A^2}, 
\label{eq:E23}
\end{equation}
which is the form of the Godbillon-Vey vector field used in the present paper. 

Alternatively if we use the usual Serret-Frenet equations involving
 $({\bf e}_1,{\bf e}_2,{\bf e}_3)$, we set
\begin{equation}
{\bf T}=\hat{\bf A}=\frac{\bf A}{A},\quad 
\boldsymbol{\omega}=\hat{\bf A}{\bf\cdot}d{\bf x}, \label{eq:E24}
\end{equation}
and a similar calculation to that in (\ref{eq:E22}) gives:
\begin{equation}
{\cal L}_{\hat{\bf A}}\left(\hat{\bf A}{\bf\cdot}d{\bf x}\right)
=-\hat{\bf A}\times 
(\nabla\times\hat{\bf A}){\bf\cdot}d{\bf x}=-\hat{\boldsymbol{\eta}}{\bf\cdot}
d{\bf x}, \label{eq:E25}
\end{equation}
where
\begin{equation}
\hat{\boldsymbol{\eta}}=\hat{\bf A}\times
(\nabla\times\hat{\bf A})\equiv -\hat{\bf A}{\bf\cdot}\nabla \hat{\bf A}. 
\label{eq:E26}
\end{equation}
In (\ref{eq:E26}) 
$\hat{\boldsymbol{\eta}}=-\kappa {\bf e}_1$
which is $-\nabla_{\bf T} {\bf T}$  where ${\bf T}\equiv \hat{\bf A}$ 
(We could have chosen $\hat{\boldsymbol{\eta}}$ to be $\nabla_{\bf T} {\bf T}$
which corresponds to the \cite{Reinhart73} formulation).
From (\ref{eq:E23}) and (\ref{eq:E26}) we obtain:
\begin{equation}
\hat{\boldsymbol{\eta}}=\boldsymbol{\eta}-{\bf w}, \label{eq:E27}
\end{equation}
where
\begin{equation}
{\bf w}=\hat{\bf A}\times(\nabla A\times \hat{\bf A})/A\equiv 
\left({\sf I}-\hat{\bf A}\hat{\bf A}\right){\bf\cdot}\nabla\ln A. \label{eq:E28}
\end{equation}
Using the fact that ${\bf A}{\bf\cdot}{\bf B}=0$ in the analysis, and using 
(\ref{eq:E27}) and (\ref{eq:E28}) we obtain:
\begin{equation}
\hat{\boldsymbol{\eta}}{\bf\cdot}\nabla\times \hat{\boldsymbol{\eta}}
=\boldsymbol{\eta}{\bf\cdot}\nabla\times\boldsymbol{\eta}+
\nabla{\bf\cdot}{\bf R}, \label{eq:E29} 
\end{equation}
where
\begin{equation}
{\bf R}=\frac{{\bf B\cdot}\nabla A}{A^3}{\bf A}-2\lambda {\bf B} 
+\lambda \nabla\ln A\times{\bf A}\quad  
\hbox{and}\quad \lambda=\frac{{\bf A\cdot}\nabla A}{A^3}. \label{eq:E30}
\end{equation}
Assuming, that $\nabla{\bf\cdot R}$ on 
the right handside of (\ref{eq:E29})
when integrated over the whole volume $V_g$ vanishes, we 
obtain the Godbillon-Vey invariant:
\begin{equation}
\hat{H}_{gv}=\int_{V_m} \hat{\boldsymbol{\eta}}{\bf\cdot}\nabla\times 
\hat{\boldsymbol{\eta}}\ d^3{\bf x}
= \int_{V_m} \boldsymbol{\eta}{\bf\cdot}\nabla\times\boldsymbol{\eta}
\ d^3{\bf x}. \label{eq:E31}
\end{equation}
The Godbillon-Vey helicity integral $\hat{H}_{gv}$  in (\ref{eq:E31}) 
is that of 
\cite{Reinhart73} and \cite{Rovenski18,Rovenski19},  
which is equivalent to the Godbillon-Vey helicity 
$H_{gv}$
used in the present paper, provided ${\bf R\cdot n}=0$ on the 
boundary $\partial V_m$ of the volume $V_m$. 
 
\appendix
\section*{Appendix F}
\setcounter{section}{6}
\setcounter{equation}{0} 
In this appendix we derive the magnetic field representation (\ref{eq:ll22})-(\ref{eq:ll23}) for ${\bf A}$ 
and ${\bf B}$ for the \cite{Low90} nonlinear force free magnetic fields. 
The condition ${\bf A\cdot B}=0$ implies that the 
Pfaffian ${\bf A\cdot}d{\bf x}$ is integrable, which in turn implies 
${\bf A}$ can be written in the form:
\begin{equation}
{\bf A}=\chi\nabla\Phi,\quad  {\bf B}=\nabla\chi\times\nabla\Phi. \label{eq:F1}
\end{equation}
From (\ref{eq:F1}) we obtain the equations:
\begin{equation}
A_r=\chi\deriv{\Phi}{r}=0,\quad \frac{\chi}{r}\deriv{\Phi}{\theta}= A_\theta, 
\quad \frac{\chi}{r\sin\theta}\deriv{\Phi}{\phi}=A_\phi, \label{eq:F2}
\end{equation}
where $A_\theta$ and $A_\phi$ are given by (\ref{eq:ll16}). The integrability 
conditions for (\ref{eq:F2}), i.e. $\Phi_{\theta\phi}=\Phi_{\phi\theta}$, 
implies that $\chi$ must satisfy the first order, linear partial differential
equation:
\begin{equation}
r\sin\theta \deriv{\chi}{\theta}-r A_\theta\deriv{\chi}{\phi}
+\chi\left[\derv{\phi}(rA_\theta)-\derv{\theta}\left(r\sin\theta A_\phi\right)
\right]=0. \label{eq:F3}
\end{equation}
The characteristics of (\ref{eq:F3}) are given by:
\begin{equation}
\frac{dr}{ds}=0,\quad \frac{d\theta}{ds}=r\sin\theta A_\phi, 
\quad \frac{d\phi}{ds}= -r A_\theta, \quad \frac{d\chi}{ds} =-W\chi, 
\label{eq:F4}
\end{equation}
where $s$ is the affine parameter along the characteristics,
\begin{equation}
W=\derv{\phi}\left(r A_\theta\right)
-\derv{\theta}\left(r\sin\theta A_\phi\right)
\equiv -r^2\sin\theta B_r
=\frac{\sin\theta (dP/d\mu)}{r^n}. \label{eq:F5}
\end{equation}
Integrating the characteristics (\ref{eq:F4}) gives the integrals:
\begin{equation}
r=c_1,\quad \phi+\frac{a}{n}\int^\mu \frac{d\mu}{(1-\mu^2)} P(\mu)^{1/n}
=c_2\equiv \xi, \quad \frac{\chi}{P(\mu)}=c_3, \label{eq:F6}
\end{equation}
where $c_1$, $c_2$, and $c_3$ are integration constants. Thus, using the theory of characteristics for first order partial differential equations, the solution of (\ref{eq:F3}) for $\chi$ have the form:
\begin{equation}
\chi=g(\xi,r) P(\mu), \quad 
\xi=\phi+\frac{a}{n}\int^\mu \frac{d\mu}{(1-\mu^2)} P(\mu)^{1/n}, 
\label{eq:F7}
\end{equation}
where $g(\xi,r)$, for the moment is an arbitrary function of $\xi$ and $r$. 

Returning to (\ref{eq:F1})-(\ref{eq:F2}) 
we require that $\Phi$ satisfy the equations:
\begin{align}
\deriv{\Phi}{r}=&\frac{1}{\chi} A_r=0, \nonumber\\
\deriv{\Phi}{\mu}=&-\frac{1}{\sin\theta}\deriv{\Phi}{\theta}=\frac{1}{g(\xi,r)}
\frac{a P(\mu)^{1/n}}{r^n (1-\mu^2)}, \nonumber\\
\deriv{\Phi}{\phi}=&\frac{1}{r^n g(\xi,r)}. \label{eq:F8}
\end{align}
Integrating (\ref{eq:F8}) gives the solution for $\Phi$ of the form:
\begin{equation}
\Phi=\int^\xi \frac{d\xi'}{G(\xi')}
\quad\hbox{where}\quad G(\xi)=r^n g(\xi,r). \label{eq:F9}
\end{equation}
In (\ref{eq:F9}) the form of $g(\xi,r)=r^{-n}G(\xi)$ is required because 
$A_r=0$ and $\partial\Phi/\partial r=0$ where $G(\xi)$ is an arbitrary
function of $\xi$. 

To summarize, the above analysis implies  
the solutions (\ref{eq:ll23})-(\ref{eq:ll24}) for ${\bf A}$, 
${\bf B}$, $\Phi$, $\chi$ and $\xi$ given in the text. 

\appendix
\section*{Appendix G}
\setcounter{section}{7}
\setcounter{equation}{0}

In this appendix we compute the curvature $\kappa$, the torsion $\tau$ 
of the normal curve to the foliation $\xi=const.$ for the \cite{Low90}
nonlinear force free magnetic field for which the magnetic vector potential
has the form:
\begin{equation}
{\bf A}=\chi\nabla\xi, \label{eq:G1}
\end{equation}
where
\begin{align}
\xi=&\phi+\frac{a}{n} \int^\mu \frac{P(\mu')^{1/n}}{1-\mu^{'2}}\ d\mu', 
\label{eq:G2}\\
\chi=&\frac{P(\mu)}{r^n}, \label{eq:G3}
\end{align}
(see (\ref{eq:ll22}) et seq.).

Using  (\ref{eq:G1})-(\ref{eq:G3}) we obtain (after some algebra, described 
below), the formulae:
\begin{align}
\kappa=&\frac{\left(\zeta^2+\sin^2\theta\right)^{1/2}}{r\sin\theta}, 
\label{eq:G4}\\
\tau=&\frac{\gamma\sin^3\theta}{(\gamma^2+1)^{1/2}  
r\left(\zeta^2+\sin^2\theta\right)}
\frac{d}{d\mu}\left(\frac{\zeta}{\sin\theta}\right), \label{eq:G5}\\
h_{21}=&-\frac{\zeta\sin\theta}{r \left(\zeta^2+\sin^2\theta\right)}
\frac{d}{d\mu}\left(\frac{\gamma\sin\theta}{(\gamma^2+1)^{1/2}}\right), 
\label{eq:G6}\\
\hat{h}_{gv}=&\kappa^2(\tau-h_{21})=\frac{1}{r^3}
\frac{d}{d\mu}\left(\frac{\gamma\zeta}{(\gamma^2+1)^{1/2}}\right) 
\label{eq:G7}
\end{align}
where
\begin{equation}
\zeta=\frac{\left(\mu+\gamma\gamma_\mu (1-\mu^2)/(\gamma^2+1)\right)}
{(\gamma^2+1)^{1/2}},\quad \gamma=\frac{a}{n} P(\mu)^{1/n}. \label{eq:G8}
\end{equation}
Here  $h_{21}$ is the component of the second fundamental form for the 
foliation $\xi=const.$ described in Appendix E 
in (\ref{eq:E13})-(\ref{eq:E14}) and (\ref{eq:E17}), and $\hat{h}_{gv}$
is the Godbillon-Vey helicity density described by the \cite{Reinhart73} 
form (\ref{eq:E17}). 
The derivation of the formulae (\ref{eq:G1})-(\ref{eq:G8}) 
are described below. 

The basis of the above results (\ref{eq:G1})-(\ref{eq:G8}) is the moving 
trihedron $({\bf T},{\bf N},{\bf B})$ describing the curve normal to 
the foliation with tangent vector ${\bf T}=\hat{\bf A}$ where 
${\bf A}$ is the magnetic vector potential (\ref{eq:ll27}) which is normal
to the foliation. Here ${\bf N}$ is the principal normal to the 
curve and ${\bf B}$ is the bi-normal to the curve. $({\bf T},{\bf N},{\bf B})$
satisfy the Serret Frenet equations (\ref{eq:E2}).  The tangent vector 
${\bf T}$ is given by:
\begin{equation}
{\bf T}=\hat{\bf A}=\frac{\nabla\xi}{|\nabla\xi|}=\frac{({\bf e}_{\phi}
-\gamma {\bf e}_\theta)}{(\gamma^2+1)^{1/2}}. \label{eq:G9}
\end{equation}
Calculating $d{\bf T}/ds={\bf T}{\bf\cdot}\nabla {\bf T}$ we obtain:
\begin{equation}
\frac{d{\bf T}}{ds}=\kappa {\bf N}=-\frac{{\bf e}_r}{r} 
-\frac{({\bf e}_\theta+\gamma {\bf e}_\phi)\zeta}{(\gamma^2+1)^{1/2}\sin\theta}, \label{eq:G10}
\end{equation}
from which we identify:
\begin{equation}
{\bf N}=-\frac{[\sin\theta {\bf e}_r+\zeta[{\bf e}_\theta+\gamma {\bf e}_\phi]
/(\gamma^2+1)^{1/2}]}
{(\zeta^2+\sin^2\theta)^{1/2}}, \label{eq:G11}
\end{equation} 
as the principal normal to the curve and 
\begin{equation}
\kappa=\frac{(\zeta^2+\sin^2\theta)^{1/2}}{r\sin\theta}, \label{eq:G12}
\end{equation}
as the principal curvature of the curve. 
The bi-normal to the curve is given by
\begin{equation}
{\bf B}={\bf T}\times{\bf N}= \frac{\left\{\zeta {\bf e}_r
-\sin\theta({\bf e}_\theta+\gamma {\bf e}_\phi)/(\gamma^2+1)^{1/2}\right\}}
{(\zeta^2+\sin^2\theta)^{1/2}}. \label{eq:G13}
\end{equation}
In the above calculations we used the formulas:
\begin{align}
\frac{d}{ds}=&{\bf T}{\bf\cdot}\nabla=\frac{1}{(\gamma^2+1)^{1/2} r\sin\theta}
\left[\derv{\phi}-\gamma\sin\theta\derv{\theta}\right], \nonumber\\
{\bf e}_r=&\left(\sin\theta\cos\phi,\sin\theta\sin\phi,\cos\theta\right), 
\nonumber\\
{\bf e}_\theta=&(\cos\theta\cos\phi,\cos\theta\sin\phi,-\sin\theta),\nonumber\\
{\bf e}_\phi=&(-\sin\phi,\cos\phi,0), \label{eq:G14}
\end{align}
The derivatives of the spherical polar unit vectors  ${\bf e}_r$, 
${\bf e}_\theta$ and ${\bf e}_\phi$ are:
\begin{align} 
\deriv{{\bf e}_r}{r}=&0,\quad \deriv{{\bf e}_r}{\theta}={\bf e}_\theta, \quad 
\deriv{{\bf e}_r}{\phi}=\sin\theta{\bf e}_\phi, \nonumber\\
\deriv{{\bf e}_\theta}{r}=&0,\quad \deriv{{\bf e}_\theta}{\theta} =-{\bf e}_r,
\quad \deriv{{\bf e}_\theta}{\phi}=\cos\theta{\bf e}_\phi, \nonumber\\
\deriv{{\bf e}_\phi}{r}=&0, \quad \deriv{{\bf e}_\phi}{\theta}=0, \quad
\deriv{{\bf e}_\phi}{\phi}=-[\sin\theta{\bf e}_r+\cos\theta{\bf e}_\theta].
\label{eq:G15}
\end{align}

The torsion $\tau$ in (\ref{eq:G5}) follows by noting that:
\begin{equation}
\tau=-{\bf N\cdot}({\bf T\cdot}\nabla {\bf B})\equiv -{\bf e}_1
{\bf\cdot}(\nabla_{{\bf e}_3} {\bf e}_2)=-\Gamma^1_{23}, \label{eq:G16}
\end{equation}
Also note that the coefficient $h_{21}$ in the \cite{Reinhart73} formula
in (\ref{eq:G6})-(\ref{eq:G7}) is given by:
\begin{equation}
h_{21}=\Gamma^3_{12}={\bf e}_3{\bf\cdot}(\nabla_{{\bf e}_2} {\bf e}_1)
={\bf T\cdot}(\nabla_{\bf B} {\bf N}). \label{eq:G17}
\end{equation}
It is straightforward to calculate the other coefficients $h_{ij}$ 
($i,j=1,2$) defining the second fundamental form of the surface.

\medskip
\medskip
\end{document}